\documentclass[journal,12pt,onecolumn]{IEEEtran}
\usepackage{graphicx,times, amsmath, amsfonts,comment}
\usepackage{amssymb,epstopdf,bm,bbm,amsthm}
\usepackage[noend]{algorithmic}

\usepackage{multirow, mathtools}

\usepackage{epstopdf, soul, color}
\usepackage{algorithm, array}

\usepackage{relsize}
\newcommand{\beq}{\begin{equation}}
\newcommand{\eeq}{\end{equation}}

\newcommand\norm[1]{\left\lVert#1\right\rVert}

\theoremstyle{plain}
\newtheorem{lemmacounter}{Theorem}
\newtheorem{lemma}[lemmacounter]{Lemma}

\begin{document}

\title{Downlink Power Allocation for CoMP-NOMA in Multi-Cell Networks}
\author{Md Shipon Ali, Ekram Hossain, Arafat Al-Dweik, and Dong In Kim\thanks{M. S. Ali and E. Hossain are with the Department of Electrical and Computer Engineering at University of Manitoba, Canada (emails: alims@myumanitoba.ca, ekram.hossain@umanitoba.ca). A. Al-Dweik is with Electrical and Computer Engineering Department, Khalifa University, UAE (email: arafat.dweik@kustar.ac.ae).  D. I. Kim is with the School of Information and Communication Engineering at Sungkyunkwan University (SKKU), Korea (email: dikim@skku.ac.kr).}}
\maketitle

\maketitle

\begin{abstract}
This work considers the problem of dynamic power allocation in the downlink of multi-cell networks, where each cell utilizes non-orthogonal multiple access (NOMA)-based resource allocation. Also, coordinated multi-point (CoMP) transmission is utilized among multiple cells to serve users experiencing severe inter-cell interference (ICI). More specifically, we consider a two-tier heterogeneous network (HetNet) consisting of a high-power macro cell underlaid with multiple low-power small cells each of which  uses the same resource block. Under this {\em CoMP-NOMA framework}, CoMP transmission is applied to a user experiencing high channel gain with multiple base stations (BSs)/cells, while NOMA is utilized to schedule CoMP and non-CoMP users over the same transmission resources, i.e., time, spectrum and space. Different CoMP-NOMA models are discussed, but focus is primarily on the joint transmission CoMP-NOMA (JT-CoMP-NOMA) model. For the JT-CoMP-NOMA model, an optimal joint power allocation problem is formulated and the solution is derived for each CoMP-set consisting of multiple cooperating BSs (i.e., CoMP BSs). To overcome the substantial computational complexity of the joint power optimization approach, we propose a distributed power optimization problem at each cooperating BS whose optimal solution is independent of the solution of other coordinating BSs. The validity of the distributed solution for the joint power optimization problem is provided and numerical performance evaluation is carried out for the proposed CoMP-NOMA models including JT-CoMP-NOMA and coordinated scheduling CoMP-NOMA (CS-CoMP-NOMA). The obtained results reveal significant gains in spectral and energy efficiency in comparison with conventional CoMP-orthogonal multiple access (CoMP-OMA) systems.
\end{abstract}


\section{Introduction}
\subsection{Preliminaries}

Due to its potential to significantly enhance the radio spectral efficiency, non-orthogonal multiple access (NOMA) is considered as a promising multiple access technology for fifth generation (5G) and beyond 5G (B5G) cellular systems \cite{saito2013}-\cite{msali2016}. The fundamental idea of NOMA is to simultaneously serve multiple users over the same transmission resources, i.e., time, spectrum and space, at the expense of inter-user interference. Although several NOMA techniques have been actively investigated over the last few years, majority of the efforts have been focused on power-domain NOMA \cite{islam2017,zte2016}, which exploits signal power diversity for each NOMA user at each NOMA receiver end. In a downlink transmission under power-domain NOMA, a base station (BS) transmitter schedules multiple users to use the same transmission resources by superposing their signals in the power domain. The superposition is performed in such a way that each NOMA user can successfully decode the desired signal by applying successive interference cancellation (SIC) technique at the corresponding receiver. In this paper, we consider {\em power-domain NOMA} and thus the term NOMA will refer to power-domain NOMA, unless it is mentioned otherwise. 

To maximize the sum-rate for downlink transmission with NOMA, BS power allocation  enables the NOMA users to perform SIC  according to the ascending order of their channel gains \cite{msali2016}. That is, prior to decoding the desired signal, each NOMA user will cancel signals of other NOMA users with  lower channel gains than the considered NOMA user. Within the same NOMA cluster, a NOMA user with higher channel gain is said to have a higher SIC order compared to a NOMA user with lower channel gain. The signals for other NOMA users with higher channel gains than the considered user will act as inter-NOMA-user interference (INUI). As a result, a cell-edge NOMA user generally experiences INUI due to signals for the cell-center users. Since with downlink transmission a NOMA user experiences the desired signal and the INUI signal over the same wireless channel, the corresponding signal-to-interference plus noise ratio (SINR) depends on the transmit power allocated to a particular user  in comparison with the sum of transmit powers allocated to the other NOMA users with higher SIC ordering \cite{msali2016}.

In a multi-cell network, co-channel downlink transmissions to cell-edge users strongly interfere with each other, which may result in low received SINR for cell-edge users. The interference could be more severe in a multi-cell heterogeneous network (HetNet) scenarios. In a co-channel HetNet, inter-cell interference (ICI) generated from a high power macro cell seriously affects on the SINR performance of the underlaid low power small cells, which may make traditional NOMA application inefficient for HetNets.
Therefore, advanced ICI management will be crucial for multi-cell downlink NOMA systems. To mitigate ICI for traditional downlink orthogonal multiple access (OMA)-based 4G cellular systems, third generation partnership project (3GPP) has adopted coordinated multi-point (CoMP) transmission technique in which multiple cells coordinately schedule/transmit to the ICI-prone users \cite{3GPP2011}. In this paper, we focus on the application of CoMP to NOMA-based multi-cell downlink transmissions in a two-tier HetNets in order to improve the network spectral efficiency. 

\subsection{Terminologies and Assumptions}

In this paper, we consider single BS at each cell (i.e., no cell sectorization) and thus the terms {\em CoMP-BS} and {\em CoMP-cell} are used interchangeably. A \textit{CoMP-set} defines the set of cells/BSs  which cooperate/coordinate to serve a user, where each coordinating cell/BS is defined as a \textit{CoMP-cell}/\textit{CoMP-BS}. The term \textit{CoMP-user} refers to a user whose desired signal is transmitted after coordination among CoMP-BSs belonging to a particular CoMP-set, while a \textit{non-CoMP-user} refers to a user whose desired signal is transmitted by only one BS without coordination with other BS(s). The term \textit{CoMP-NOMA} is used to indicate the application of CoMP in a NOMA system in which a CoMP-BS forms a NOMA cluster by including both CoMP-users and non-CoMP-users, according to the applied CoMP scheme. In addition, \textit{cluster-head} of a particular NOMA cluster defines a user who can cancel all INUI due to other users in the considered NOMA cluster.

For multi-cell downlink transmission,  a two-tier HetNet is considered with full frequency reuse where one macro cell is underlaid with several small cells (i.e., all cells use same resource block). 
The mathematical derivations and corresponding numerical analysis performed in this paper consider a particular resource block consisting of time and spectrum resource, e.g., resource block in an LTE system. A Rayleigh fading radio channel is considered where fading gain is assumed to be flat over a considered NOMA resource block. Single transmit and receive antennas are considered at both user and BS ends. 
NOMA clusters within each cell use orthogonal resource blocks (i.e., no inter-NOMA-cluster interference within a single cell). In the NOMA clusters served by different CoMP-BSs within a CoMP-set, a common CoMP-user is served by the CoMP-BSs using the same resource block(s)  . 

\subsection{Existing Research on CoMP-NOMA}

In spite of tremendous research interests in NOMA, very few existing works investigate NOMA  in multi-cell networks. 
Here, we briefly review the existing research on resource allocation and ICI mitigation/modeling in downlink multi-cell NOMA networks. 
In \cite{z.zhang2017}, the outage probability and achievable data rates are derived for uplink and downlink NOMA in a multi-cell homogeneous network. 
A joint problem for spectrum allocation and power control is formulated in \cite{j.zhao2017} for downlink in HetNets, where a many-to-one matching game is utilized for spectrum allocation and a non-convex optimization problem is formulated for power control. 
In \cite{fu2017}, under a particular rate constraint, the authors formulate a distributed optimization approach for sum transmit power minimization among a number of BSs in a downlink multi-cell network.

By utilizing CoMP technology among multiple cells, a downlink CoMP-NOMA model is studied in \cite{shin2016}. The authors discuss the application of various CoMP schemes in a multi-cell homogeneous network in which each cell utilizes $2$-user NOMA for downlink transmission.
In \cite{choi2014}, Alamouti code is utilized for joint downlink transmission to a cell-edge user under a CoMP framework consisting of two cells in a homogeneous network. In \cite{choi2014}, a CoMP-user forms a NOMA cluster with a non-CoMP-user, i.e., forms $2$-user NOMA clusters at each coordinating cell. 
Also, an application of NOMA in a downlink CoMP transmission scenario can be found in \cite{tian2016}, where NOMA is opportunistically used for a group of cell-edge users receiving CoMP transmission from multiple coordinating cells. The authors derive the outage probability for the proposed opportunistic NOMA system by considering a joint multi-cell power allocation among the CoMP-users. 

The authors in \cite{bey2016} also study a CoMP-NOMA system for downlink transmission and  propose a suboptimal scheduling strategy for NOMA users under CoMP transmission. 
In \cite{tian2017}, a downlink CoMP-NOMA system is utilized for the purpose of relaying a signal to a remote user who is unable to receive direct transmission from any coordinating cell. 
The authors  derive the outage probability for their proposed CoMP-NOMA system by considering a fixed power allocation strategy. In \cite{yang2017}, dynamic power control is used for multi-cell downlink NOMA for sum-power minimization and sum-rate maximization. The authors consider CoMP transmissions from the cells in a homogeneous network, where each cell considers two users in each NOMA cluster. 
On the other hand, by utilizing massive multiple-input multiple-output (MIMO)-enhanced macro cell and NOMA-enhanced small cell, a hybrid HetNet system is studied in \cite{liu2017}, where the authors derive the coverage probability of NOMA-enhanced small cells under the impact of MIMO-enabled macro cell. 
The problem of  resource allocation and sum-rate optimization for NOMA in a downlink multi-cell MIMO network are investigated in \cite{han2014}-\cite{n.yen2017}.

\subsection{Motivation and Key Contributions}
Although dynamic power allocation is critical to achieve performance gain in a NOMA system, most of the exiting research on CoMP-NOMA  considers fixed power allocation strategies and/or single-cell scenarios. 
Moreover, to be useful for practical scenarios,  performance modeling, analysis and optimization of NOMA should consider multi-cell scenarios including HetNets. Since the performance of a multi-cell network is prone to ICI, application of CoMP would be highly desirable, particularly when NOMA is used. Motivated by these, we investigate the dynamic power allocation problem for sum-rate maximization in downlink CoMP-NOMA in a multi-cell scenarios under minimum rate constraints for users in a NOMA cluster. 

In the proposed CoMP-NOMA model, each BS can form a NOMA cluster by including one/multiple CoMP-user(s) and one/multiple non-CoMP-user(s). The system categorizes users into CoMP-user and non-CoMP-user according to their received SINR. The CoMP-set is determined for each CoMP-user which in turn yields the number of CoMP-users in each CoMP-set. Within a CoMP-set, each coordinating cell, i.e., a CoMP-cell, forms NOMA cluster consisting of their non-CoMP-users and CoMP-users, based on the applied CoMP scheme. After forming NOMA clusters, dynamic power allocation is performed for each NOMA cluster at each CoMP-cell. 

Fig. \ref{fig:sm} demonstrates the proposed CoMP-NOMA model for a multi-cell two-tier HetNet scenario consisting of one high power macro base station (MBS), which is denoted as an eNB in LTE terminology, underlaid by two low power small cell base stations (SBSs). 
The figure shows two CoMP-sets: one $3$-BS CoMP-set and one $2$-BS CoMP-set. 
In the $3$-BS CoMP-set, UE$_{1,m}$, UE$_{1,s1}$, and UE$_{1,s2}$ are the non-CoMP-user equipment (UE)\footnote{Throughout the paper, we will use the terms ``user" and ``UE" interchangeably.} served only by MBS, SBS-1 and SBS-2, respectively, without any coordination among the CoMP-BSs, while CoMP-UEs UE$_{1,ms1s2}$ and UE$_{2,ms1s2}$ are served by all three CoMP-BSs in a coordinated manner. Similarly, in the $2$-BS CoMP-set, UE$_{2,m}$ and UE$_{2,s1}$ are the non-CoMP-UEs served only by eNB and SBS-1, respectively, while CoMP-UE UE$_{1,ms1}$ is served by both the CoMP-BSs, i.e., eNB and SBS-1. Note that Fig. \ref{fig:sm} demonstrates the joint transmission (JT)-CoMP-NOMA scenario where the CoMP-UEs receive multiple transmissions from the CoMP-BSs. However, signal transmissions to a CoMP-UE depends on the applied CoMP scheme which will be discussed in  detail in Section III. 
In \cite{msali_comp_mag}, we presented the concept of downlink CoMP-NOMA in a homogeneous multi-cell network scenario.

\begin{figure}[h]
\begin{center}
	\includegraphics[width=3.50 in]{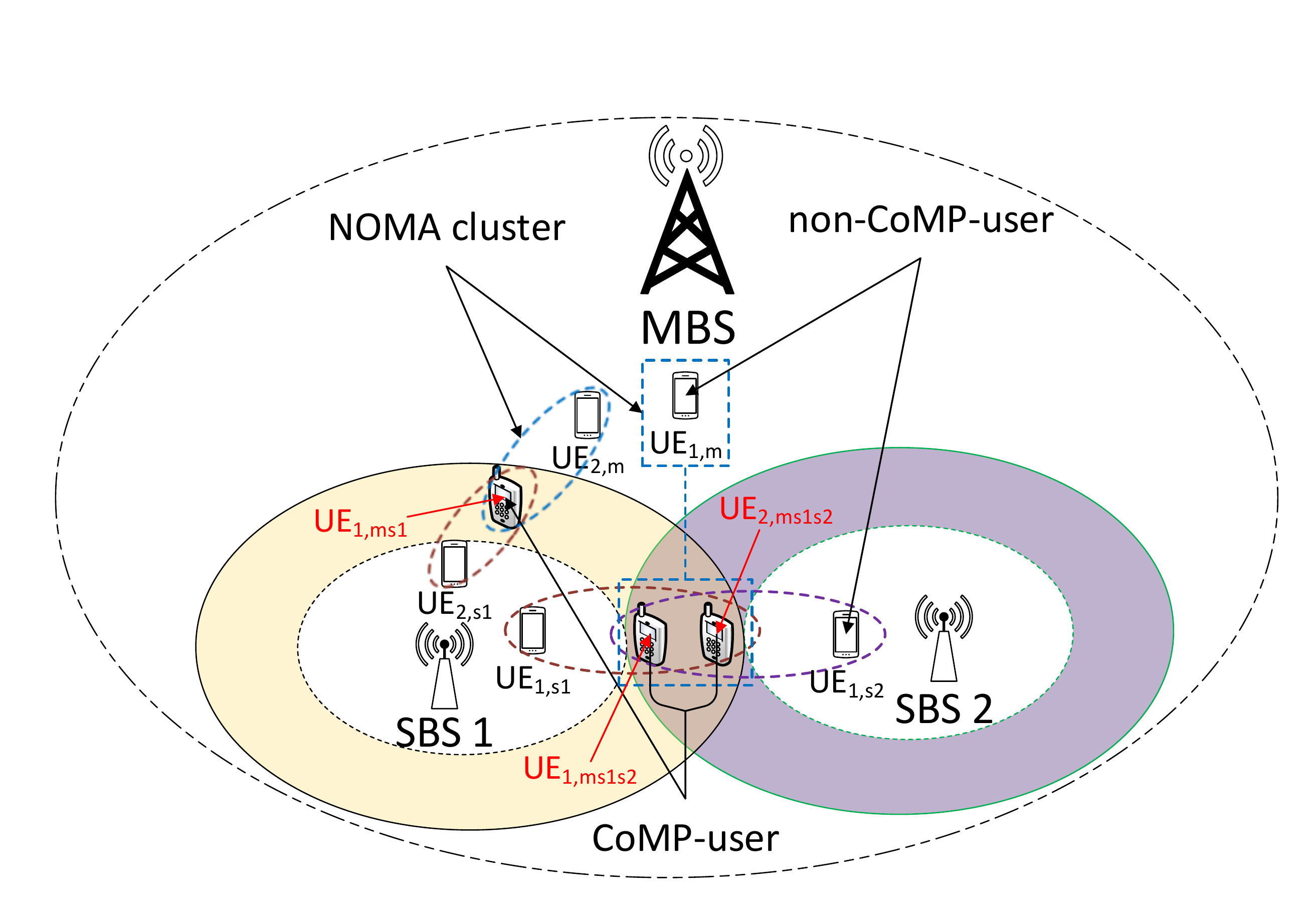}
	\caption{An illustration of the proposed downlink JT-CoMP-NOMA model in a two-tier HetNet.}
	\label{fig:sm}
 \end{center}
\end{figure} 

The key contributions of this paper can be summarized as follows: 

\begin{itemize}
\item Formulate a convex power optimization problem for sum-rate maximization in a single-cell downlink NOMA system under constrained  minimum rate requirements for users in a NOMA cluster and discuss the necessary conditions for global optimality. This serves as a basis to formulate and solve the power allocation problem for CoMP-NOMA in a multi-cell/HetNet scenarios.

\item Provide a novel system and signal model for downlink CoMP-NOMA in a multi-cell network (a two-tier HetNet scenario in particular). We derive the achievable rate formulas for CoMP-UEs and non-CoMP-UEs considering $2$-cell and $3$-cell CoMP-sets. 

\item A joint power optimization problem across multiple cells is formulated for CoMP-NOMA sum-rate maximization considering both $2$-cell and $3$-cell CoMP-sets. The solution of the joint power optimization problem is also provided. 

\item A low-complexity distributed power optimization approach for each CoMP-cell is presented, which avoids the high computational complexity of the joint power optimization problem involving all the CoMP-cells belonging to a CoMP-set.

\item Finally, provide a comprehensive  performance evaluation of the proposed CoMP-NOMA system. Numerical results demonstrate the gain in spectral and energy efficiency due to the proposed CoMP-NOMA model in comparison with traditional CoMP-OMA systems.

\end{itemize}

\subsection{Paper Organization}
The rest of the paper is organized as follows: Section~II presents the power optimization problem formulation and solution for the downlink in a single-cell  NOMA system. Section~III discusses various CoMP schemes and their applicability in the downlink of NOMA-based multi-cell networks. Also, it presents the system model, assumptions and achievable rate formula for the proposed downlink JT-CoMP-NOMA model in a two-tier HetNet. The joint and distributed power optimization approaches  for the proposed downlink JT-CoMP-NOMA model are presented in  Sections~IV and V, respectively. Section~VI provides the necessary conditions to validate the distributed solution for the joint power optimization problem and also provides insights into the energy efficiency performance of CoMP-NOMA system. Section VII evaluates the performance of the proposed solutions numerically. Finally, Section~VIII concludes the paper.

\section{Optimal Power Allocation for Single-cell Downlink NOMA}

Consider a single-cell network (i.e., no ICI) with $M \geq 2$ users having distinct channel gains which form an $M$-user NOMA cluster for downlink transmission. Also, assume that the SIC ordering for this NOMA cluster follows the users' indices, i.e., signal for UE$_1$ is decoded first, signal for UE$_2$ is decoded second, and so on. Thus, UE$_1$ can decode its desired signal by treating all other users' signals as interference (i.e., INUI), while UE$_M$ can decode its desired signal after canceling all INUI signals by applying SIC techniques. In such a downlink NOMA cluster, the achievable rate for any user $i$ can be written as follows: 
\begin{align}
R_i = B \log_2\left[1+\frac{p_i \gamma_i}{\sum\limits_{j = i+1}^{M} p_j\gamma_i+1}\right], \, \forall i = 1,2,\cdots,M
\label{scn_1}
\end{align}
where $\gamma_i$ denotes channel power gain at the receiver with a normalized noise power for user $i$, and $p_i$ denotes the transmission power for user $i$. The radio channel is assumed to be flat fading Rayleigh channel over NOMA bandwidth $B$ for each resource block. 
For successful SIC operation at the receiver end(s), the necessary conditions for power allocation are: 
\begin{align}
p_i \gamma_{k} - \sum\limits_{j=i+1}^{M} p_j\gamma_{k} - 1 \geq \theta, \enspace \substack{\forall i = 1,2,\cdots,M-1 \\  \forall k = i, i+1, \cdots, M}. 
\label{scn_2}
\end{align} 
The term $\theta$ denotes the minimum difference between the decoded signal power and the non-decoded INUI signal plus noise power \cite{msali2016}. This minimum difference is required to decode a NOMA user's signal in the presence of signals (plus noise) for other NOMA users with higher SIC ordering\footnote{User $i$ having a higher SIC order than user $j$ means that user $i$ can cancel INUI from user $j$ before decoding the desired signal, while user $j$ cannot cancel INUI from user $i$.}. To maximize the  sum-rate in a NOMA cluster, the SIC ordering needs to follow the ascending order of the NOMA users' channel gains. That is, the aforementioned SIC ordering will provide maximum sum-rate if the channel gains for the NOMA users are as follows: $\gamma_M > \gamma_{M-1}>\cdots> \gamma_1$. In such an optimal scenario, the SIC constraints in \eqref{scn_2} could be simplified as  \cite{msali2016}
\begin{align}
p_i \gamma_i - \sum\limits_{j=i+1}^{n} p_j \gamma_i - 1\geq \theta, \, \, \forall i = 1,2,\cdots,M-1 \label{scn_2_1}.
\end{align}
It is worth noting that, for an $M$-user NOMA with any particular SIC ordering, the number of SIC constraints will be reduced to exactly $(M-1)$, and it can be simplified by following a procedure similar to [8, eq. (5)].
 
Now, the power optimization problem for sum-rate maximization over unit NOMA bandwidth in an $M$-user downlink NOMA having channel gains $\gamma_M > \gamma_{M-1}>\cdots> \gamma_1$ can be formulated as follows: 
\begin{align} 
\underset{\bm {\mathbf{p}}}{\text{max}} \enspace \sum\limits_{i = 1}^{M}\log_2\left[1+\frac{p_i \gamma_i}{\sum\limits_{j = i+1}^{M} p_j\gamma_i + 1}\right]
\label{scn_3}
\end{align}
subject to: 
\begin{align*} 
&\bm{\mathbf{C1:}}~\sum\limits_{i = 1}^{M}p_i \leq p_t \nonumber \\ 
&\bm{\mathbf{C2:}}~ \log_2\left[1+\frac{p_i \gamma_i}{\sum\limits_{j = i+1}^{M} p_j\gamma_i + 1}\right] \geq R_i, \, \forall i = 1,2,\cdots, M \nonumber\\
&\bm{\mathbf{C3:}}~ p_{i}\gamma_{i+1} - \sum\limits_{j = i+1}^{M}p_{j}\gamma_{i+1} -1 \geq \theta, \, \forall i = 1, 2,\cdots,M-1
\end{align*}
where $\bm{\mathbf{C1}}$ and $\bm{\mathbf{C3}}$ represent NOMA power budget and SIC constraints, respectively, while $\bm{\mathbf{C2}}$ represents the constraint on each NOMA user's rate requirement. In this paper, each NOMA user's minimum rate requirement is considered to be equal to their achievable rate in an OMA system. For example, in the case of equal spectrum and power allocation among $M$ users in an OMA system,  $R_i = \frac{B}{M} \log_2\big(1 + p_t\gamma_i\big), \, \forall i = 1,2, \cdots, M$. 


\begin{lemma}
For an $M$-user ($M\geq 2$) downlink NOMA system with ascending channel gain-based SIC ordering, the sum-rate is a strictly concave function of the transmit powers for NOMA users.
\end{lemma}

\begin{proof}
See \textbf{Appendix A}.
\end{proof}

\begin{lemma}
All the constraints in downlink NOMA sum-rate maximization problem formulated in \eqref{scn_3} are convex/concave function of the transmit powers for NOMA users.
\end{lemma}

\begin{proof}
For an $M$-user downlink NOMA sum-rate maximization problem, there are $2M$ constraints: one constraint for the power budget  ($\bm{\mathbf{C1}}$), $M-1$ constraints for the SIC  ($\bm{\mathbf{C3}}$), and $M$ other constraints for the rate requirements of $M$ NOMA users ($\bm{\mathbf{C2}}$). In the optimization problem in \eqref{scn_3},  the power budget constraint $\bm{\mathbf{C1}}$ and the SIC  constraints $\bm{\mathbf{C3}}$ all are affine functions of NOMA users' power allocation, and thus they are convex/concave. In the following, we prove that the constraints $\bm{\mathbf{C2}}$ are also convex/concave function of NOMA users' power allocation. 

Consider the  rate constraint for user $i$ as follows: 
\begin{align}
\log_2\left[1+\frac{p_i \gamma_i}{\sum\limits_{j = i+1}^{M} p_j\gamma_i + 1}\right] \geq R_i, \, \forall i = 1,2,\cdots, M.
\label{L2_1}
\end{align}
By letting $\zeta = (2^{R_i} -1)$,  \eqref{L2_1} can be written as 
\begin{align}
p_i \gamma_i - \zeta\sum\limits_{j = i+1}^{M} p_j\gamma_i - \zeta \geq 0
\label{L2_2}
\end{align}
which is an affine function of $\bm{\mathbf{p}} = [p_1, p_2, \cdots, p_M]^{T}$. Because all NOMA users' rate requirements follow a similar form, then they are all affine functions, i.e., convex/concave function of $\bm{\mathbf{p}}$.
\end{proof}

According to \textbf{Lemma 1} and \textbf{Lemma 2}, the sum-rate maximization problem in \eqref{scn_3} is a convex optimization problem. Also, \textbf{Lemma 2} indicates that all the inequality constraints are affine, and therefore, the Slater's condition holds [25, p. 227].    
Therefore, according to Slater's theorem, the \textit{KKT conditions are both necessary and sufficient for the optimal solution to the problem formulated in \eqref{scn_3}}. Hence, the closed-form solutions for the downlink NOMA proposed in \cite{msali2016} will be the global optimal solution for the problem in \eqref{scn_3}.

\section{Downlink CoMP-NOMA in a Two-tier HetNet}
This section discusses the application of various CoMP schemes in downlink NOMA and provide power allocation strategy for each COMP-NOMA model for downlink transmission in a two-tier HetNet. Then, the focus will be on the joint transmission  CoMP-NOMA (JT-CoMP-NOMA) downlink  model and present the achievable rate formulas. Although the analysis is carried out  for HetNets, it is also valid for homogeneous networks. 

\subsection{CoMP-NOMA in Multi-cell Downlink Networks}

In \cite{msali_comp_mag}, different CoMP-NOMA models for multi-cell networks are discussed. A key observation from \cite{msali_comp_mag} is that the power allocation for coordinated scheduling CoMP-NOMA (CS-CoMP-NOMA) system in downlink transmission is exactly the same as the power allocation in single-cell NOMA, when the NOMA clusters in different CoMP-cells, which belong to a CoMP-set, utilize different resource blocks. In downlink CS-CoMP-NOMA, a CoMP-UE receives transmission from only one CoMP-BS in a channel which is orthogonal to other channels where the other CoMP users receive transmissions from other CoMP-BSs belonging to the same CoMP-set. Thus, under downlink CS-CoMP-NOMA transmission, the NOMA clusters do not experience ICI, and hence, power allocation in downlink CS-CoMP-NOMA is similar to that in a single-cell NOMA system as presented in Section II. 

The downlink transmission for dynamic point selection (DPS)-CoMP scheme is similar to that under CS-CoMP scheme except that the former requires the CoMP UEs' data to be available at each CoMP-BS. Both schemes select one CoMP-BS to transmit to a CoMP-UE at a time instant while other CoMP-BSs belonging to the same CoMP-set does not use that spectrum \cite{3GPP2011}.  Thus, at each scheduling time interval, the power allocation for downlink DPS-CoMP NOMA is also same as the power allocation for single-cell NOMA. 
It is also noted that, under CS-CoMP/DPS-CoMP scheme, a CoMP-BS which does not transmit to a CoMP-UE, may use the same spectrum to serve a cell-center user with a low power to minimize ICI.
 Due to the possible infeasibility of MIMO precoding and/or decoding,  coordinated beamforming (CB)-CoMP scheme is usually not applicable to downlink NOMA systems~\cite{msali_comp_mag}. 
 
On the contrary to the other CoMP schemes, joint transmission (JT)-CoMP scheme enables multiple CoMP-BSs to transmit to a single CoMP-user simultaneously, and thus the power allocation in JT-CoMP-NOMA is different from that in single-cell NOMA.  In the following sections, we will focus on the dynamic power allocation for JT-CoMP-NOMA in multi-cell downlink HetNets scenarios.

\subsection{JT-CoMP NOMA System Model}
Consider a HetNet  consisting of a single high power eNB underlaid by $X$ low transmit power SBSs from the set $\mathcal{X} = \{1,2,\cdots,X\}$. For downlink transmission in this HetNet, each BS uses NOMA to schedule its users, while JT-CoMP transmission by multiple CoMP-BSs is applied to serve ICI-prone CoMP-users. We also assume that the number of coordinating BSs is no more than three. Therefore, the number of $2$-BS CoMP-set $Y = X$ and the number of $3$-BS CoMP-set $Z \leq$ $X \choose 2$. It is also assumed that $\mathcal{Y} = \{1,2,\cdots,Y\}$ is the set of $2$-cell CoMP-set, and $\mathcal{Z} = \{1,2,\cdots,Z\}$ is the set of three-cell CoMP-set. 

Based on the received SINR, users are categorized into the following sets: set of  non-CoMP-UEs served by the eNB, denoted as $\mathcal{U_M} = \{1,2,\cdots,U_M\}$; set of  non-CoMP-UEs served by the $x$-th SBS, denoted as $\mathcal{U_{S_X}} = \{1,2,\cdots,U_{S_x}\}$; set of CoMP-UEs under a $2$-BS CoMP-set formed by the eNB and the $x$-th SBS, denoted as $\mathcal{U_{MS_X}} = \{1,2,\cdots,U_{MS_x}\}$; and  set of CoMP-UEs under a $3$-BS CoMP-set formed by the eNB, $x$-th SBS, and $x^\prime$-th SBS, denoted as $\mathcal{U_{MS_XS_{X^\prime}}} = \{1,2,\cdots,U_{MS_xS_{x^\prime}}\}, \forall x \neq x^\prime$.

Under this JT-CoMP-NOMA model, a CoMP-UE  receives multiple transmissions from the CoMP-BSs belonging to a CoMP-set and forms distinct NOMA clusters with non-CoMP-UEs served by each of these CoMP-BSs of the considered CoMP-set. On the other hand, a non-CoMP-UE receives her desired signal only from the BS it is associated with, which can be a member of only one NOMA cluster. 
It is worth noting that each CoMP-UE simultaneously forms different NOMA clusters at multiple CoMP-cells belonging to a CoMP-set, and hence, its SIC ordering at each NOMA cluster should follow \textbf{Lemma~3}.

\begin{lemma}
The SIC ordering for a CoMP-user will be same in all NOMA clusters formed at each CoMP-cell of a CoMP-set.
\end{lemma}

\begin{proof}
See [24, p. 3].
\end{proof}

In the proposed CoMP-NOMA model, each NOMA cluster must include at least one non-CoMP-UE and the SIC ordering for the non-CoMP-UEs will be followed by the SIC ordering for the CoMP-UEs, i.e., each non-CoMP-UE can cancel INUI due to the CoMP-UEs while none of the CoMP-UE can cancel INUI due to the non-CoMP-UEs. 
This  SIC ordering for CoMP-UEs is due to their locations at the cell-edge area (and hence low channel gain). Note that  this SIC ordering also  reduces the overhead of SIC significantly when compared to the opposite SIC ordering \cite{msali_comp_mag}. 
In addition, we adopt the downlink NOMA system proposed in \cite{msali2016} where the users are initially grouped into clusters and then dynamic power allocation is performed for each NOMA cluster.  

For the sake of simplicity, we assume that all CoMP-UEs in a particular CoMP-set will be grouped into a single NOMA cluster formed at each CoMP-cell. Also, the minimum number of NOMA clusters in a particular cell/BS is assumed to be equal to the number of CoMP-sets in which the considered BS is a member. 

 
In the following, the dynamic power allocation method is presented for our proposed CoMP-NOMA model considering a particular CoMP-set.
 
\subsection{Achievable Data Rate in JT-CoMP-NOMA System}

Consider a particular $2$-BS CoMP-set $y\in \mathcal{Y}$, where the number of non-CoMP-UEs in the MBS and SBS  $x$ are $\Phi_{y,m}$ and $\Phi_{y,s_x}$, respectively, and the number of CoMP-UEs is $\Phi_{y,ms_x}$. 
Within this $y$-th $2$-BS CoMP-set, the set of non-CoMP-UEs $i \in \{1,\cdots,\Phi_{y,m} \}$ in the MBS, the set of  non-CoMP-UEs  $j \in \{1,\cdots,\Phi_{y,s_x} \}$ in the $x$-th SBS, and the set of CoMP-UEs $k \in \{1,\cdots,\Phi_{y,ms_x} \}$ all are assumed to follow SIC ordering according to their subscripts. For the non-CoMP-UEs at each cell, the subscripts follow the ascending order of their channel gains with their respective serving BS. On the other hand,  for the CoMP-UEs, their subscripts may not follow  the ascending order of their channel gains with any particular CoMP-BS. 

Under the above assumptions, the achievable rate (over unit resource block) for  non-CoMP-UE $i$, which is in CoMP-set $y$ and served by the eNB, can be expressed as
\begin{align}
R_i =\log_2\left[1+\frac{p_i^{(m)} \gamma_i^{(m)}}{ \sum\limits_{i^\prime = i+1}^{\Phi_{y,m}} p_{i^\prime}^{(m)} \gamma_i^{(m)} + \norm{\bm{\mathbf{p}}_{j}^{(s_x)}}_1 \gamma_i^{(s_x)} + 1}\right]
\label{mcn_1}
\end{align}
where $\gamma_i^{(m)}$ and $\gamma_i^{(s_x)}$ represent $i$-th non-CoMP-UE's channel gain with eNB (desired channel) and SBS $x$ (ICI channel), respectively. 
The term $p_i^{(m)}$ denotes the downlink transmit power allocated by the eNB for the non-CoMP-UE $i$ (i.e., transmit power for the desired signal), $\sum_{i^\prime = i+1}^{\Phi_{y,m}} p_{i^\prime}^{(n)}$ represents the downlink transmit power for other non-CoMP-UEs served by the  eNB who are in the same NOMA cluster as user $i$ but has higher SIC ordering than user $i$ 
(i.e., transmit power corresponding to INUI signal).
Also, $\norm{\bm{\mathbf{p}}_{j}^{(s_x)}}_1 = \sum_{j = 1}^{\Phi_{y,s_x}} p_j$ represents the downlink transmit power for non-CoMP-UEs in SBS $x$ which form different NOMA clusters with the common CoMP-UEs in the same NOMA cluster as user $i$ (i.e., transmit power for ICI signal).

Over the same unit resource block used in \eqref{mcn_1}, the achievable rate for a non-CoMP-UE $j$ in SBS $x$ can be expressed as
\begin{align}
R_{j} =\log_2\left[1+\frac{p_{j}^{(s_x)} \gamma_{j}^{(s_x)}}{ \sum\limits_{j^\prime = j+1}^{\Phi_{y,s_x}} p_{j^\prime}^{(s_x)} \gamma_{j}^{(s_x)} + \norm{\bm{\mathbf{p}}_{i}^{(m)}}_1 \gamma_{j}^{(m)} + 1}\right]
\label{mcn_2}
\end{align}
where $\gamma_j^{(s_x)}$ and $\gamma_j^{(m)}$ represent the $j$-th non-CoMP-UE's channel gain with SBS $x$ (desired channel) and eNB (ICI channel), respectively. The term $p_j^{(s_x)}$ and $\sum_{j^\prime = j+1}^{\Phi_{y,s_x}} p_{j^\prime}^{(s_x)}$ are similarly defined as for \eqref{mcn_1}, but for  non-CoMP-UE $j$ in SBS $x$ in this case. Here, $\norm{\bm{\mathbf{p}}_{i}^{(m)}}_1 = \sum_{i = 1}^{\Phi_{y,m}} p_i$ represents the downlink transmit power for  non-CoMP-UEs $i$ served by the MBS, which form different NOMA clusters with the common CoMP-UEs in the same cluster as user $j$ (i.e., transmit power for ICI signal).

Also, over the same unit resource block used in \eqref{mcn_1} and \eqref{mcn_2}, the achievable rate for a CoMP-UE $k$ can be expressed as
\begin{align}
&R_{k} = \log_2\left[1+\frac{\bm{\mathbf{p}}_{k}^T \bm{\mathbf{\gamma}}_{k}}{ \sum\limits_{k^\prime = k+1}^{\Phi_{y,ms_x}} \bm{\mathbf{p}}_{k^\prime}^T \bm{\mathbf{\gamma}}_{k}+ \norm{\bm{\mathbf{p}}_{i}^{(m)}}_1 \gamma_{k}^{(m)}+ \norm{\bm{\mathbf{p}}_{j}^{(s_x)}}_1 \gamma_{k}^{(s_x)} + 1}\right]
\label{mcn_3}
\end{align}
where $\gamma_k^{(m)}$ and $\gamma_k^{(s_k)}$ represent $k$-th CoMP-UE's channel gain with eNB and SBS $x$, respectively, and CoMP-UE $k$ receives desired signal over both channels. 
The term $\bm{\mathbf{p}}_{k^\prime}^T \bm{\mathbf{\gamma}}_{k} = p_{k}^{(m)} \gamma_{k}^{(m)} + p_{k}^{(s_x)} \gamma_{k}^{(s_x)}$ represents the desired signal for CoMP-UE $k$ jointly transmitted from both CoMP-BSs (eNB and SBS $x$), where $\bm{\mathbf{p}}_{k} = [p_{k}^{(m)}, p_{k}^{(s_x)}]^T$, $ \bm{\mathbf{\gamma}}_{k} = [\gamma_{k}^{(m)}, \gamma_{k}^{(s_x)}]^T$ and $\bm{\mathbf{p}}_{k}^T$ indicates the transpose of $\bm{\mathbf{p}}_{k}$. The term $\sum_{k^\prime = k+1}^{\Phi_{y,ms_x}} \norm{p_{k^\prime}\gamma_{k}}_1$ represents INUI due to other CoMP-UEs of CoMP-set $y$ which form NOMA cluster with CoMP-UE $k$ but have higher SIC ordering than UE $k$. The  terms $\norm{\bm{\mathbf{p}}_{j}^{(s_x)}}_1$ and $\norm{\bm{\mathbf{p}}_{i}^{(m)}}_1$ are similarly defined as in \eqref{mcn_1} and \eqref{mcn_2}, respectively, but act as INUI from the non-CoMP-UEs of SBS $x$ and MBS, respectively, which form different NOMA clusters with CoMP-UE $k$. Note that a CoMP-UE experiences INUI from both the CoMP-BSs due to its inclusion in both the NOMA clusters and having a lower SIC ordering than the non-CoMP-UEs.

\section{Sum-Rate Maximization in JT-CoMP NOMA: Joint Power Optimization Approach} 
Consider the same $2$-BS CoMP-set $y$ that are considered in Section III. Therefore, the number of non-CoMP-UEs served by the eNB, non-CoMP-UEs in the $x$-th SBS,  and CoMP-UEs in CoMP-set $y$ are $\Phi_{y,m}$, $\Phi_{y,s_x}$ and $\Phi_{y,ms_x}$, respectively. Also, assume that $p_t^{(m)}$ and $p_t^{(s_x)}$ are the NOMA power budget at the eNB and the $x$-th SBS, respectively, for CoMP-set $y$. For CoMP-set $y$, if $R_{i}^\prime$, $R_{j}^\prime$, and $R_{k}^\prime$ are the rate requirements for non-CoMP-UE $i$ served by the eNB,  non-CoMP-UE $j$ in the $x$-th SBS, and common CoMP-UE $k$, respectively, then the sum-rate maximization problem for the $2$-BS CoMP-set $y$ can be formulated as follows\footnote{Note that, similar formulations can be done for higher order CoMP-sets such as for 3-BS CoMP-sets.}: 

\begin{align} 
\underset{\bm{\mathbf{p}}^{(m)}, \bm{\mathbf{p}}^{(s_x)}}{\text{max}}~ \sum\limits_{i = 1}^{\Phi_{y,m}}R_i + \sum\limits_{j = 1}^{\Phi_{y,s_x}}R_j + \sum\limits_{k = 1}^{\Phi_{y,ms_x}}R_k
\label{mcn_op_1}
\end{align}
subject to: 
\begin{align*} 
&\bm{\mathbf{C1:}}~\sum\limits_{i = 1}^{\Phi_{y,m}}p_{i}^{(m)} + \sum\limits_{k = 1}^{\Phi_{y,ms_x}}p_{k}^{(m)} \leq p_t^{(m)} \nonumber \\ 
&\bm{\mathbf{C2:}}~\sum\limits_{j = 1}^{\Phi_{y,s_x}}p_{j}^{(s_x)} + \sum\limits_{k = 1}^{\Phi_{y,ms_x}}p_{k}^{(s_x)} \leq p_t^{(s_x)} \nonumber \\ 
&\bm{\mathbf{C3:}}~ R_k \geq R_k^\prime, \quad \forall k = 1,\cdots, \Phi_{y,ms_x} \nonumber \\
&\bm{\mathbf{C4:}}~ R_i \geq R_i^\prime, \quad \,\,\forall i = 1,\cdots, \Phi_{y,m} \nonumber\\
&\bm{\mathbf{C5:}}~ R_j \geq R_j^\prime, \quad \,\forall j = 1,\cdots, \Phi_{y,s_x} \nonumber\\
&\bm{\mathbf{C6:}}~ \bm{\mathbf{p}}_{k}^T \bm{\mathbf{\gamma}}_{l} - \sum\limits_{k^\prime = k+1}^{\Phi_{y,ms_x}}\bm{\mathbf{p}}_{k^\prime}^T \bm{\mathbf{\gamma}}_{l} - \norm{\bm{\mathbf{p}}_i^{(m)}}_1 \gamma_{l}^{(m)} - \, \norm{\bm{\mathbf{p}}_j^{(s_x)}}_1 \gamma_{l}^{(s_x)} - 1 \geq \theta,  \,\Big\{\substack{\forall k = 1,\cdots,(\Phi_{y,ms_x}-1) \\  \forall l = k, \cdots, \Phi_{y,ms_x} \quad\enspace\,\,} \nonumber \\ 
\end{align*}
\begin{align*}
&\bm{\mathbf{C7:}}~ p_{i}^{(m)}\gamma_{i+1}^{(m)} - \sum\limits_{i^\prime = i+1}^{\Phi_{y,m}}p_{i^\prime}^{(m)}\gamma_{i+1}^{(m)} - \norm{\bm{\mathbf{p}}_j^{(s_x)}}_1 \gamma_{i+1}^{(s_x)} - 1 \geq \theta, \, \forall i = 1,\cdots,\left(\Phi_{y,m}-1\right)  \\
&\bm{\mathbf{C8:}}~ p_{j}^{(s_x)}\gamma_{j+1}^{(s_x)} - \sum\limits_{j^\prime = j+1}^{\Phi_{y,s_x}}p_{j^\prime}^{(s_x)}\gamma_{i+1}^{(s_x)} - \norm{\bm{\mathbf{p}}_i^{(m)}}_1 \gamma_{j+1}^{(m)} - 1 \geq \theta, \,  \forall j = 1, \cdots,\left(\Phi_{y,s_x}-1\right)  
\end{align*}
where $R_i$, $R_j$ and $R_k$ are defined in \eqref{mcn_1}, \eqref{mcn_2} and \eqref{mcn_3}, respectively, Also, the vector $\bm{\mathbf{p}}_{k}$, $\bm{\mathbf{p}}_{k^\prime}$, $\bm{\mathbf{p}}_{i}^{(m)}$, $\bm{\mathbf{p}}_{j}^{(s_x)}$ and $\bm{\mathbf{\gamma}}_{l}$ are defined similarly as in \eqref{mcn_3}. 
The optimization variable vector $\bm{\mathbf{p}}^{(m)}$ and $\bm{\mathbf{p}}^{(s_x)}$ are defined as $\bm{\mathbf{p}}^{(m)} = [p_1^{(m)}, \cdots, p_{\Phi_{y,m}}^{(m)}, p_{\Phi_{y,m}+1}^{(m)}, \cdots, p_{\Phi_{y,m}+\Phi_{y,ms_x}}^{(m)}]^T$ and $\bm{\mathbf{p}}^{(s_x)} = [p_1^{(s_x)}, \cdots, p_{\Phi_{y,s_x}}^{(s_x)}, p_{\Phi_{y,s_x}+1}^{(s_x)}, \cdots, p_{\Phi_{y,s_x}+\Phi_{y,ms_x}}^{(s_x)}]^T$.
Constraints $\bm{\mathbf{C1}}$ and $\bm{\mathbf{C2}}$ represent the power budget for NOMA cluster at eNB and SBS $x$, respectively, under CoMP-set $y$. Constraints $\bm{\mathbf{C3}}$, $\bm{\mathbf{C4}}$, and $\bm{\mathbf{C5}}$ represent the individual rate requirements for CoMP-UEs,  non-CoMP-UEs served by the eNB, and  non-CoMP-UEs served by the SBS, respectively, which form NOMA cluster(s) in CoMP-set $y$. The SIC application requirements for CoMP-UEs and  non-CoMP-UEs in CoMP-set $y$ are represented by constraints $\bm{\mathbf{C6}}$, $\bm{\mathbf{C7}}$, and $\bm{\mathbf{C8}}$, respectively. Since SIC ordering for the non-CoMP-UEs follows their ascending channel gain order, we use SIC constraint \eqref{scn_2_1} for non-CoMP-UEs. On the other hand, the SIC constraint in \eqref{scn_2} is used for CoMP-UEs as they may not follow ascending channel gain order at both the CoMP-BSs simultaneously. 


The sum-rate maximization problem for JT-CoMP-NOMA formulated in \eqref{mcn_op_1} is a joint optimization problem among the CoMP-BSs. Each CoMP-BS needs to solve problem \eqref{mcn_op_1} by considering all the possible solutions for other CoMP-BSs of the considered CoMP-set. That is, for each feasible power allocation solution at SBS $x$, the eNB maximizes \eqref{mcn_op_1} by optimizing $\bm{\mathbf{p}}^{(m)}$ under constraints $\bm{\mathbf{C1}},\bm{\mathbf{C3}},\bm{\mathbf{C4}},\bm{\mathbf{C6}}$, and $\bm{\mathbf{C7}}$. For each feasible $\bm{\mathbf{p}}^{(s_x)}$, i.e., constant $\bm{\mathbf{p}}^{(s_x)}$, the optimization of $\bm{\mathbf{p}}^{(m)}$ in \eqref{mcn_op_1} is similar to the optimization of $\bm{\mathbf{p}}$ in \eqref{scn_3}, and thus the solution\footnote{This solution could be a global or local maxima depending on the SIC ordering of non-CoMP-UEs and CoMP-UEs.} can be obtained by using the KKT conditions \cite{msali2016}. 
On the other hand, SBS $x$ maximizes \eqref{mcn_op_1} by optimizing $\bm{\mathbf{p}}^{(s_x)}$ under constraints $\bm{\mathbf{C2}},\bm{\mathbf{C3}},\bm{\mathbf{C5}},\bm{\mathbf{C6}}$, and $\bm{\mathbf{C8}}$, for each feasible solution at the MBS. In addition, since the CoMP-UEs' SIC ordering may not follow their ascending channel gain order, the joint optimization should be solved for all possible SIC ordering for the CoMP-UEs. 
Therefore, solving the joint power optimization problem \eqref{mcn_op_1} to obtain the optimal solution will incur substantial computational complexity. To reduce the complexity, in the following section, we will present a method for distributed power optimization at each CoMP-BS which is independent of the power allocations at other CoMP-BSs.

\section{Distributed Power Allocation in JT-CoMP NOMA}
This section provides a suboptimal solution for the optimization problem \eqref{mcn_op_1} by developing another optimization problem that can be solved distributively at each CoMP-BS without considering power allocations at the other CoMP-BSs. In the proposed solution, each CoMP-BS optimizes its power allocation for the NOMA cluster without considering ICI and the impact of CoMP transmission. However, note that, due to JT-CoMP transmission, all of the  CoMP-BSs simultaneously transmit the same message signal to the CoMP-users by using the power allocation solution presented in Section II. 


Consider the same CoMP-set $y$ that is considered in Section IV. In CoMP-set $y$, the NOMA cluster served by the MBS contains $\Phi_{y,m}$ non-CoMP-UEs and $\Phi_{y,ms_x}$ CoMP-UEs, and the NOMA cluster served by SBS $x$ contains $\Phi_{y,s_x}$ non-CoMP-UEs and $\Phi_{y,ms_x}$ CoMP-UEs. In CoMP-set $y$, let $\Psi_{y,m} = \Phi_{y,m} + \Phi_{y,ms_x}$ denote the total number of users in the NOMA cluster served by the MBS and  $\Psi_{y,s_x} = \Phi_{y,s_x} + \Phi_{y,ms_x}$  denotes the total number of users in the NOMA cluster served by SBS $x$. For any CoMP-BS $n \in \{m,s_x\}$ in CoMP-set $y$, the achievable rate formula for a NOMA user $l \in \{1,2,\cdots, \Psi_{y,n}\}$ (either a non-CoMP-UE or a CoMP-UE) can be written as
\begin{align}
\hat{R}_l =\log_2\left[1+\frac{p_l^{(n)} \gamma_l^{(n)}}{\sum\limits_{l^\prime = l+1}^{\Psi_{y,n}} p_{l^\prime}^{(n)} \gamma_l^{(n)} + 1}\right]
\label{s_mcn_1}
\end{align}
where $\gamma_l^{(n)}$ represents channel gain for NOMA user $l$ with serving BS $n$. The terms $p_l^{(n)}$ and $\sum_{l^\prime = l+1}^{\Psi_{y,n}} p_{l^\prime}^{(n)}$, respectively, denote the downlink transmit power for user $l$ (i.e., transmit power for desired signal) and downlink transmit power for other users in BS $n$ which form NOMA cluster with user $l$ but have higher SIC ordering than user $l$ (i.e., transmit power for INUI signal). Then, the distributed sum-rate maximization problem for CoMP-BS $n\in \{m,s_x\}$ can be formulated as
\begin{align} 
\underset{\bm{ \mathbf{p}}^{(n)}}{\text{max}}~ \sum\limits_{l = 1}^{\Psi_{y,n}} \hat{R}_l 
\label{s_mcn_op_1}
\end{align}
subject to: 
\begin{align*} 
&\bm{\mathbf{\hat{C}1:}}~\sum\limits_{l = 1}^{\Psi_{y,n}}p_{l}^{(n)} \leq p_t^{(n)} \nonumber \\
&\bm{\mathbf{\hat{C}2:}}~ \hat{R}_l \geq R_l^\prime, \, \forall l = 1,2,\cdots, \Psi_{y,n} \nonumber\\ 
&\bm{\mathbf{\hat{C}3:}}~ p_{l}^{(n)}\gamma_{q}^{(n)} - \sum\limits_{l^\prime = l+1}^{\Psi_{y,n}}p_{l^\prime}^{(n)}\gamma_{q}^{(n)} - 1 \geq \theta, \,  \Big\{\substack{\forall l = 1,2,\cdots,\left(\Psi_{y,n}-1\right) \\  \forall q = l,(l+1), \cdots, \Psi_{y,n}}  
\end{align*}
where $\bm{\mathbf{\hat{C}1}}$ represents the power budget constraint for NOMA cluster at CoMP-BS $n$, while $\bm{\mathbf{\hat{C}2}}$ and $\bm{\mathbf{\hat{C}3}}$, respectively, denote each of the NOMA user's rate requirement and SIC requirement in a NOMA cluster at CoMP-BS $n$ which belongs to CoMP-set $y$. The optimization problem in \eqref{s_mcn_op_1} is exactly the same as the optimization problem in \eqref{scn_3}. 
Thus, the KKT optimization method \cite{msali2016} can be utilized to obtain the optimal solution (see footnote 2) for the distributed optimization approach. Note that, for $3$-BS CoMP-sets, the distributed power optimization approach is the same as \eqref{s_mcn_op_1}, but it should be solved at each of the three CoMP-BSs.

\section{Validity of the  Distributed Power Optimization Approach and Insights}

\subsection{Validity of the DPO Approach for JT-CoMP-NOMA}

A solution for the distributed power optimization (DPO) approach in \eqref{s_mcn_op_1} will be a feasible solution for the joint power optimization (JPO) approach in \eqref{mcn_op_1} 
if all the constraints in \eqref{s_mcn_op_1} belong to the feasible solution region of \eqref{mcn_op_1}. Note that, the KKT optimality model in \cite{msali2016} can be used to obtain the optimal solution for both the JPO and DPO approaches. 
 According to \cite{msali2016}, the optimal power allocation solution for a downlink NOMA always provides the minimum power to all NOMA users except the cluster-head; however, that minimum power must satisfy each user's rate requirement and the SIC requirement. After satisfying all the requirements for non-cluster-head
NOMA users, the rest of the available power budget is allocated to the cluster-head. Under this solution technique, the necessary conditions, under which a solution obtained using the DPO approach will be a feasible solution for the JPO approach, are as follows:

\subsubsection{Power Budget Constraint}
The NOMA power budget constraint for the DPO approach \eqref{s_mcn_op_1} and the JPO approach \eqref{mcn_op_1} are exactly the same, which can be easily verified by setting $n = m$, i.e., constraints $\bm{\mathbf{\hat{C}1}}$ and $\bm{\mathbf{C1}}$  are same, 
and by setting $n = s_x$, i.e., constraints $\bm{\mathbf{\hat{C}1}}$ and $\bm{\mathbf{C2}}$
 are same. Note that $\Psi_{y,n} = \Phi_{y,n} + \Phi_{y,ms_x},\, \forall n\in \{m,s_x\}$.

\subsubsection{Data Rate and SIC Constraints}
In JT-CoMP-NOMA, each CoMP-BS forms a NOMA cluster by incorporating both CoMP-UEs and non-CoMP-UEs, where the non-CoMP-UEs are different at each CoMP-BS but CoMP-UEs are the same with similar SIC ordering at each CoMP-BS belonging to a particular CoMP-set. 
The achievable rate formula for a CoMP-UE and a non-CoMP-UE are different in the JPO approach \eqref{mcn_op_1}, while the achievable rate formulas for all NOMA users (either CoMP-UE or non-CoMP-UE) are the same in the DPO approach \eqref{s_mcn_op_1}. 
The conditions under which the NOMA users' rate and SIC constraints under the DPO approach \eqref{s_mcn_op_1} will satisfy the respective constraints under the JPO approach \eqref{mcn_op_1} are as follows:
\begin{itemize}

\item[$\bullet$] If the ICI for non-CoMP-UEs is negligible, then the non-CoMP-UEs' achievable rate formula under the JPO and DPO approaches are exactly the same, i.e., \eqref{mcn_1} and \eqref{s_mcn_1} are similar when the ICI component in \eqref{mcn_1} is negligible.

\item[$\bullet$] If the ICI is not negligible for non-CoMP-UEs, then with the DPO approach, the minimum rate requirement for the non-CoMP-UEs (except the cluster-head) cannot be met. By introducing an offset ICI into the achievable rate formula for a non-CoMP-UE under DPO approach, the  minimum rate requirement can be easily satisfied. The offset ICI, denoted as $\hat{I}_{l}^{n^\prime}$, to a NOMA user $l$ for interfering BS $n^\prime$ can be expressed as 
\begin{align}
\hat{I}_{l}^{n^\prime} = \left(p_t^{(n^\prime)}-\sum\limits_{k\in \mathcal{K}} p_k^{(n^\prime)}\right)\gamma_{l}^{(n^\prime)}
\label{offsetICI_1}
\end{align}
where $\gamma_{l}^{(n^\prime)}$ is the power gain of interfering channel normalized with respect to noise power, $\mathcal{K}$ is the set of CoMP-UEs, $p_k^{(n^\prime)}$ is the distributively allocated power to CoMP-UE $k$ from CoMP-BS $n^\prime$, and $p_t^{(n^\prime)}$ denotes the NOMA power budget at CoMP-BS $n^\prime$. Since each CoMP-BS individually meets the rate requirement for a common CoMP-UE, it can easily determine $\hat{I}_{l}^{n^\prime}$ for any CoMP-BS $n^\prime$ if the power budget $p_t^{(n^\prime)}$ is known to all CoMP-BSs $n$.

\item[$\bullet$] According to \textbf{Lemma 4} below, the rate constraint for a CoMP-UE under the DPO approach always satisfies the constraint under the JPO approach.

\item[$\bullet$] In a JT-CoMP-NOMA model, since the CoMP-UEs receive simultaneous transmissions from all CoMP-BSs, their corresponding SIC constraints under the DPO approach satisfy the SIC constraints under the JPO approach (similar to \textbf{Lemma 4}). For a non-CoMP-UE, by adding the offset ICI of \eqref{offsetICI_1} into the SIC constraint in \eqref{s_mcn_op_1}, the corresponding SIC constraints in \eqref{mcn_op_1} can be satisfied. It can be easily verified by noting the constraints. 

\end{itemize}

\begin{lemma}
The achievable rate for a CoMP-UE under the DPO approach \eqref{s_mcn_op_1} is higher than the rate that can be achieved under the JPO approach \eqref{mcn_op_1}. The data rate of a CoMP-UE under the DPO approach will be exactly equal to that under the JPO approach if the noise power in their SINR expressions is divided by the number of CoMP-BSs, i.e., if \eqref{s_mcn_1} is modified as
\begin{align}
\hat{R}_l =\log_2\left[1+\frac{p_l^{(n)} \gamma_l^{(n)}}{\sum\limits_{l^\prime = l+1}^{\Psi_{y,n}} p_{l^\prime}^{(n)} \gamma_l^{(n)} + 1/\underline{N}}\right], \, \forall l\in \mathcal{K}
\label{LL_4}
\end{align} 
where $\underline{N}$ is the number of CoMP-BS and $\mathcal{K}$ is the set of common CoMP-users.
\end{lemma}

\begin{proof}
See \textbf{Appendix B}.
\end{proof}

According to \textbf{Lemma~3}, the SIC ordering for each CoMP-UE should be the same for the distributed power allocation approach in \eqref{s_mcn_op_1} used at each CoMP-BS.
 On the other hand, \textbf{Lemma~1} restricts the convexity of problem \eqref{s_mcn_op_1} for only ascending channel gain-based SIC ordering. Therefore, the optimal SIC ordering may not be possible at all the CoMP-BSs simultaneously. For maximum sum-rate across all the CoMP-BSs, all the possible combinations of CoMP-UEs' channel gain orders should be checked exhaustively. In the section on numerical results, the rate performance of different decoding orders for the CoMP-UEs will be investigated and we will explain the insights generated through the results. Note that, at each CoMP-BS, non-CoMP-UEs  have a higher SIC ordering than CoMP-UEs.

\subsection{Insights on Energy Efficiency Performance of CoMP-NOMA}


JT-CoMP is an effective rate improvement technique for ICI-prone users in downlink OMA systems. Under an OMA system, signal transmitted to a CoMP-UE by each CoMP-BS does not contain any interference and thus the resultant SINR significantly improves (i.e., SINR = SNR). On the other hand, with NOMA, signal transmitted to a CoMP-UE by each CoMP-BS contains INUI, and thus the resultant SINR may not improve significantly. Under distributed power allocation, \textbf{Lemma~4} shows that each CoMP-BS individually needs to meet the same rate requirement for each CoMP-UE. In addition, under the JT-CoMP-NOMA model, a CoMP-UE has a lower SIC ordering. Thus, according to the achievable rate formula in \eqref{s_mcn_1}, the SINR for each CoMP-UE in each NOMA cluster can be approximated to the ratio of desired transmit power to INUI power 
 gain, i.e., 
\begin{align}
\frac{p_l^{(n)} \gamma_l^{(n)}}{\sum\limits_{l^\prime = l+1}^{\Psi_{y,n}} p_{l^\prime}^{(n)} \gamma_l^{(n)}+ 1} \approx \frac{p_l^{(n)}}{\sum\limits_{l^\prime = l+1}^{\Psi_{y,n}} p_{l^\prime}^{(n)}}
\label{EE_1}
\end{align}
if $\sum_{l^\prime = l+1}^{\Psi_{y,n}} p_{l^\prime}^{(n)} \gamma_l^{(n)}+ 1 \approx \sum_{l^\prime = l+1}^{\Psi_{y,n}} p_{l^\prime}^{(n)} \gamma_l^{(n)},\, \forall l \in \Phi_{y,ms_x}$.
Therefore, the average energy requirement for transmitting each bit could be significantly high for downlink JT-CoMP-NOMA.

With JT-CoMP-NOMA, since each CoMP-BS individually meets the rate requirements for CoMP-UEs, we can apply either DPS-CoMP or CS-CoMP instead of JT-CoMP to achieve similar rate of JT-CoMP-NOMA. In such a DPS/CS-CoMP scheme, all of the CoMP-BSs use the same resource block for NOMA cluster under a particular CoMP-set but only one CoMP-BS includes the CoMP-UEs into her NOMA cluster, while the others form NOMA clusters including only non-CoMP-UEs. However, under the proposed JT-CoMP-NOMA model, since all of the CoMP-BSs  use the same resource block, the CoMP-BSs only transmitting to the non-CoMP-UEs could cause significant interference to  the CoMP-UEs. By controlling the transmit power at the CoMP-BSs  transmitting only to non-CoMP-UEs, the interference caused to the CoMP-UEs under the DPS/SC-CoMP could be minimized. In such a case, the sum-rate under DPS/SC-CoMP-NOMA may decrease to some extent in comparison with JT-CoMP-NOMA.


\section{Numerical Analysis}

In this section,  we examine the spectral efficiency (SE) for JT-CoMP-NOMA under the proposed JPO and DPO approaches and also compare the results with those for JT-CoMP-OMA. We also examine the energy efficiency (EE) gain for the proposed JT-CoMP-NOMA system over their OMA counterparts. The SE is measured in bits/sec/Hz while EE is measured in Mb/J. The SE and EE for a CoMP-set is taken as a summation over all the users served by that CoMP-set, while the SE and EE for a CoMP-BS is taken as a summation over all the users at the considered BS. 
All the numerical results are generated using MATLAB.

\subsection{Simulation Assumptions}

We consider a full frequency reuse in a two-tier HetNet consisting of one macrocell underlaid by two partially overlapped 
small cells, which is similar to Fig.\ref{fig:sm}. Each BS in the corresponding cell is located at the center of the corresponding circular coverage area. 
 We manually solve both optimization problems, i.e., JPO and DPO, and obtain the optimal solutions by utilizing the KKT optimality conditions \cite{msali2016}. 
 For JPO, the optimal power allocation is obtained by searching ten thousand points over the feasible power range in each CoMP-BS, i.e., search step size is $p_t/10000$. For DPO, we use \eqref{s_mcn_1} and \eqref{LL_4} to obtain the achievable data rate for non-CoMP users and CoMP users, respectively. In each NOMA cluster, we use one/multiple CoMP-UE but the number of non-CoMP-UEs is restricted to one, and thus no offset ICI is used under the DPO approach. Note that an offset ICI expressed in \eqref{offsetICI_1} needs to be used for all non-CoMP-UEs except the cluster-head under the DPO approach.

\begin{table} [h]
\centering
\caption{simulation parameters.}
\label{parameters}
\begin{tabular}{|c|c|}
\hline
Parameter         & Value                          \\
\hline \hline
eNB to SBS distance      & $0.75$ Km                       \\
\hline
SBS to SBS distance      & $0.30$ Km                       \\
\hline
System effective bandwidth, $B$     & $8.64$ MHz                       \\
\hline
Bandwidth of one resource block & $180$ kHz                             \\
\hline 
Transmit power budget at eNB, $p_t^{(m)}$     & $46$ dBm                       \\
\hline
Transmit power budget at SBS $x$, $p_t^{(s_x)}$     & $25$ dBm                       \\
\hline
Number of antennas at BS/UE end   & $1$                     \\
\hline 
Antenna gain at BS/UE end    & $0$ dBi                      \\
\hline
SIC detection threshold, $\theta$    & $10$ dBm  \\
\hline
Receiver noise spectral density density, $N_0$    & $-169$ dBm/Hz  \\
\hline                     
\end{tabular}
\end{table}

In the simulations, the wireless channel is assumed to experience only path-loss. The path-loss for a user at a distance $d$ Km from BS is modeled as: $128.1 + 37.6\log_{10}(d)$.  User categorization (i.e., selection of CoMP-UE and non-CoMP-UE) and NOMA clustering, which are done based on the distances of the users from the different BSs, are assumed to be known a priori before power allocation is performed. Since power allocation is performed in order to mitigate the effects of large-scale fading, only path-loss is considered. ICI from a BS outside the CoMP-set is ignored (i.e., assumed to be part of the noise power).  

The simulation parameters are considered in accordance of the 3GPP standards and the major  parameters are shown in Table \ref{parameters}. The number of resource block blocks (RB) is considered to be equal to the number of users in a NOMA cluster. In case of orthogonal multiple access (OMA), the transmit power and spectrum RB are assumed to be equally allocated among the OMA users. Note that EE is calculated by averaging sum-rate over the transmit power. Perfect channel state information (CSI) is assumed to be available at the BS ends. For convenience, we use the term $n$:$m$:$k$ to indicate an JT-CoMP NOMA system consisting of $n$ CoMP-BS (a single eNB and $n-1$ SBS) each of which having individual NOMA cluster of $m$ users and each cluster contains common $k$ number of CoMP-UEs with the same SIC ordering at each cluster and $m-k$ non-CoMP-UEs who are different for different NOMA clusters.

\subsection{Spectral Efficiency Performance}

\begin{figure*}[h]
\begin{center}
	\includegraphics[width=8.5 in]{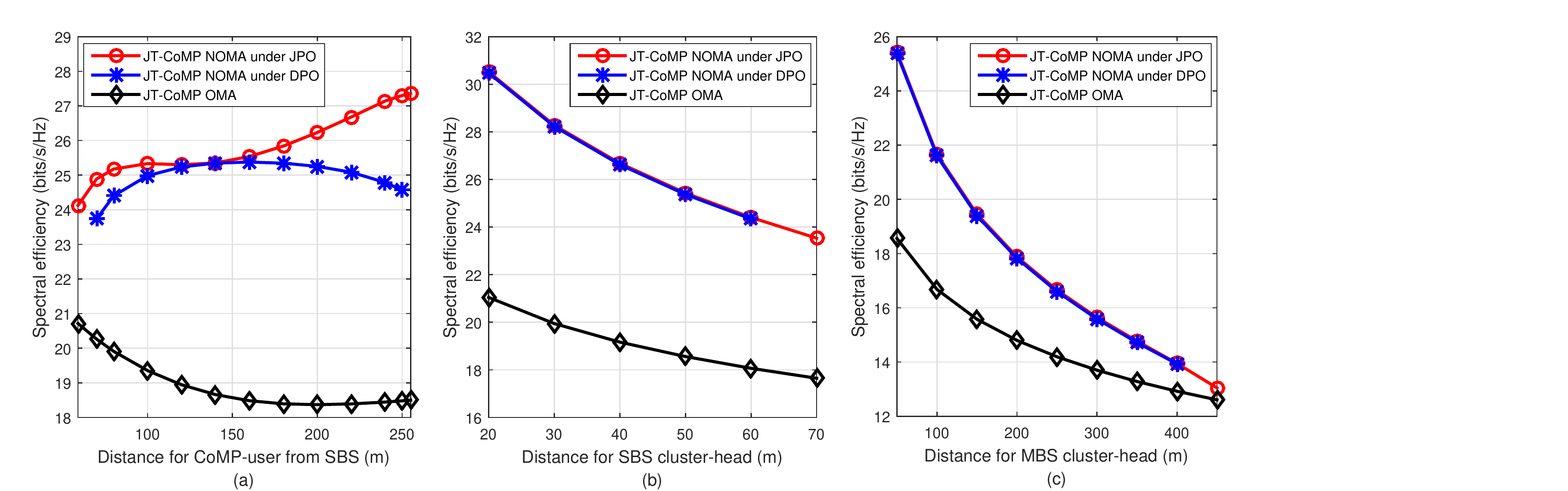}
	\caption{Spectral efficiency for JT-CoMP-NOMA model $2$:$2$:$1$ and JT-CoMP-OMA: (a) each non-CoMP-UE is at a distance of $50$ m  from its serving BS, (b) non-CoMP-UE of macrocell is at a distance of $50$ m  from the eNB and the CoMP-UE is at a distance of $150$ m  from SBS, (c) non-CoMP-UE of SBS is at a distance of $50$ m from SBS and the CoMP-UE is at a distance of $150$ m from SBS.}
	\label{fig:ii}
 \end{center}
\end{figure*} 
\begin{figure*}
\begin{center}
	\includegraphics[width=7.7 in]{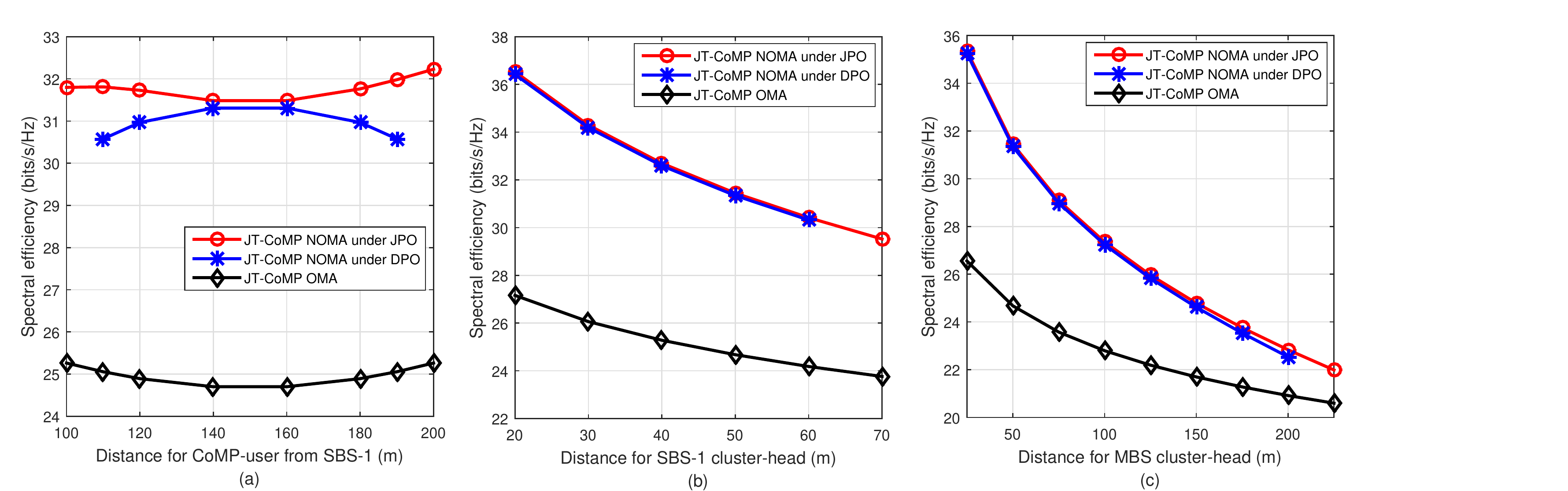}
	\caption{Spectral efficiency for JT-CoMP-NOMA model $3$:$2$:$1$ and JT-CoMP-OMA system: (a) each non-CoMP-UE is at a distance of $50$ m from its serving BS, (b) the non-CoMP-UEs of macrocell  and SBS-2 are at a distance of $50$ m from their serving BSs, while the CoMP-UE is at distance of $150$ m from SBS-1, (c) non-CoMP-UEs of SBS-1 and SBS-2 are at a distance of $50$ m from their serving BSs, while the CoMP-UE is at a distance of $150$ m  from the SBSs.}
	\label{fig:iii}
 \end{center}
 \begin{center}
	\includegraphics[width=7.45 in]{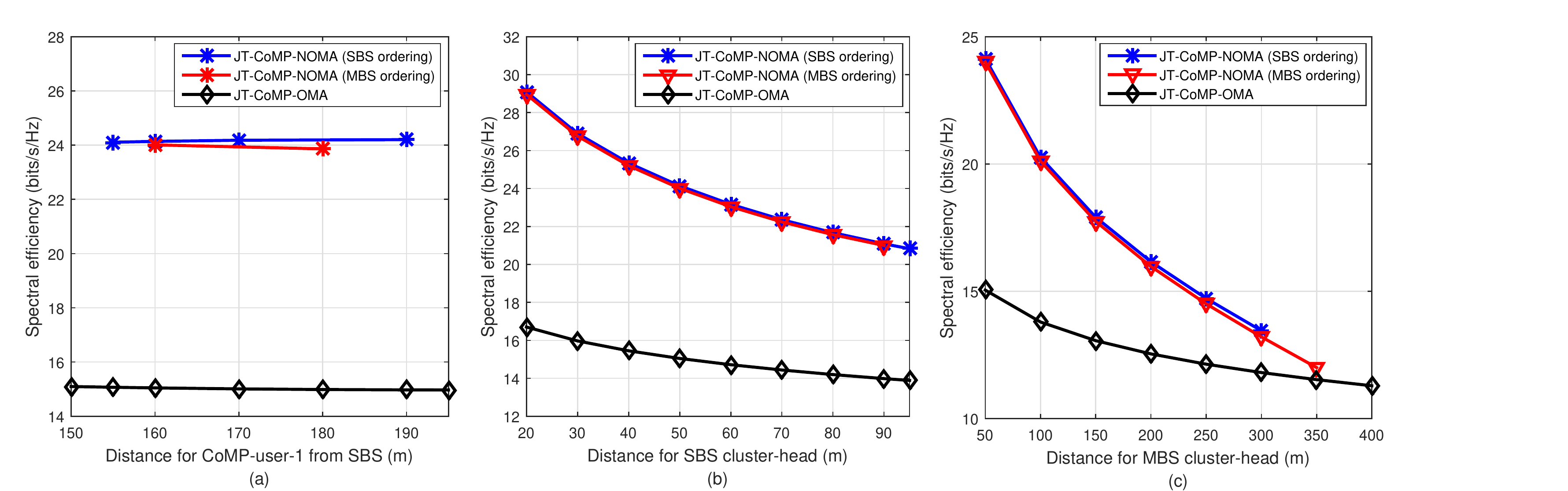}
	\caption{Spectral efficiency for JT-CoMP-NOMA model $2$:$3$:$2$ and JT-CoMP-OMA system: (a) each non-CoMP-UE is at a distance of $50$ m  from its serving BS  and CoMP-UE-2 is at a distance of $150$ m from SBS, (b) non-CoMP-UE of macrocell is at a distance of $50$ m  from eNB, while CoMP-UE-1 is at a distance of $100$ m and CoMP-UE-2 is at a distance of $150$ m from SBS, (c) non-CoMP-UE of SBS is at a distance of $50$ m  from SBS while the CoMP-UEs are at distance as in (b).}
	\label{fig:iv}
 \end{center}
\end{figure*}

The SE performance of the proposed JT-CoMP-NOMA system is analyzed by simulating three different JT-CoMP-NOMA models: $2$:$2$:$1$, $3$:$2$:$1$ and $2$:$3$:$2$, and the corresponding simulation results are shown in Fig. \ref{fig:ii}, Fig. \ref{fig:iii}, and Fig. \ref{fig:iv}. Each individual figure also consists of three different results demonstrating SE performance in terms of distance for CoMP-UE,  non-CoMP-UEs served by SBS, and non-CoMP-UEs served by the eNB, respectively in (a), (b) and (c). While the distance for a non-CoMP-UE is measured with respect to its serving BS, the distances of the CoMP-UEs  are measured with respect to the SBS. In the case of multiple SBSs in a CoMP-set (e.g., $3$-BS CoMP-set), SBS-$1$ is considered to determine the distance for the CoMP-UEs.

All three figures, i.e., Fig. \ref{fig:ii},  Fig. \ref{fig:iii}, and Fig. \ref{fig:iv}, demonstrate a significant improvement of SE gain for JT-CoMP NOMA systems over their OMA counterparts. However, if we compare the results in Fig. \ref{fig:ii}(a) and Fig. \ref{fig:iii}(a), we observe that the range of distance for CoMP-user (which in turns determines the range of variation of CoMP-user's channel power gain) to form NOMA cluster with non-CoMP-user is reduced as the number of BS in a CoMP-set increases. As the number of BSs in a CoMP-set increases, the desired signal power for CoMP-UE also increases additively under JT-CoMP-OMA. On the other hand, for a CoMP-UE under 
JT-CoMP-NOMA, the desired signal power and INUI power both additively increase as the number of BSs in a CoMP-set increases. Also, the data rate for a CoMP-UE in a NOMA cluster is proportional to the ratio of power allocated to that CoMP-UE and the power allocated to all other NOMA users which have higher SIC ordering. Thus, to achieve the same data rate that a CoMP-UE would achieve with OMA  in Fig. \ref{fig:iii}(a), the transmission power available for a CoMP-UE in each NOMA cluster could be insufficient.

Fig. \ref{fig:ii}(b), \ref{fig:iii}(b), \ref{fig:ii}(c) and \ref{fig:iii}(c) demonstrate a significantly large range of non-CoMP-UE's channel variation to form NOMA cluster with CoMP-UEs. This is because, the non-CoMP-UE in each NOMA cluster is the cluster-head and has very high channel power gain (due to closer distance to their serving BS), and thus until their channel gains decrease significantly, they can achieve a higher data rate even-though they are allocated a small portion of the transmit power. 

Fig. \ref{fig:ii}(a) and Fig. \ref{fig:iii}(a) also show the performance gap between the proposed JPO approach and the low-complexity DPO approach. It can be found that the performance gap is very small around the middle points (120-160 m in Fig. \ref{fig:ii}(a) and 140-160 m in Fig. \ref{fig:iii}(a)) over the range of CoMP-UEs' channel variations. In JT-CoMP, each CoMP-BS individually needs to meet the rate requirement for the common CoMP-UEs. On the other hand, under the JPO approach, to allocate power to a CoMP-UE, a CoMP-BS considers all the possible combinations of NOMA power allocation used by other CoMP-BSs, and thus obtain global optimal power allocation over all CoMP-BSs. On the other hand, under the DPO approach, each CoMP-BS independently allocates power to CoMP-UE and non-CoMP-NOMA users, and thus obtains a locally optimal power allocation at each CoMP-BS. However, the resultant SE performance gap is not significant and it can be considered as the cost of the  complexity reduction of JT-COMP-NOMA. Note that the computational complexity of the JPO approach is of $O(2^{n\times k})$, where $n$ and $k$ are the number of CoMP-BSs and CoMP-UEs, respectively, within a CoMP-set.

Fig. \ref{fig:iv} illustrates the SE gain for the proposed JT-CoMP-NOMA model under DPO approach over JT-CoMP-OMA.  Fig. \ref{fig:iv}(a) shows that the range for non-CoMP-UEs' channel variation decreases to much lower values in comparison with Fig. \ref{fig:ii} (a) and \ref{fig:iii}(a). Since two CoMP-UEs are included in each NOMA cluster, the rate requirement for CoMP-users may not facilitate sufficient flexibility for their channel variation under JT-CoMP-NOMA. Moreover, as the power budget of eNB is significantly higher than that  of SBS, the CoMP-users' SIC ordering according to their ascending channel gain with SBS provides more flexibility for JT-CoMP-NOMA, which is indicated in Fig. \ref{fig:iv}(a).

\subsection{Energy Efficiency Performance}

\begin{figure}[h]
\begin{center}
	\hspace*{9 em}\includegraphics[width=9 in]{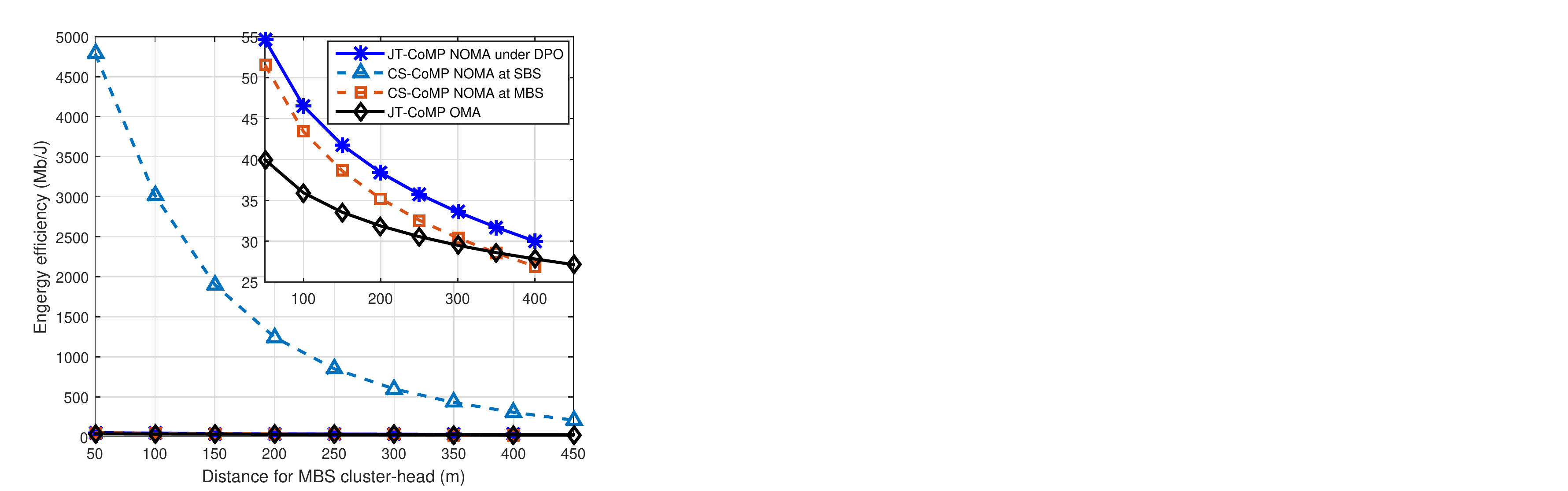}
	\caption{Energy efficiency for JT-CoMP-NOMA model $2$:$2$:$1$, corresponding CS-CoMP-NOMA and JT-CoMP-OMA. Non-CoMP-UE of SBS is at a distance of $50$ m  from SBS while the CoMP-UE is at a distance of $100$ m from SBS.}
	\label{fig:v}
 \end{center}
\end{figure} 
\begin{figure}[h]
\begin{center}
	\hspace*{6.7 em}\includegraphics[width=9.5 in]{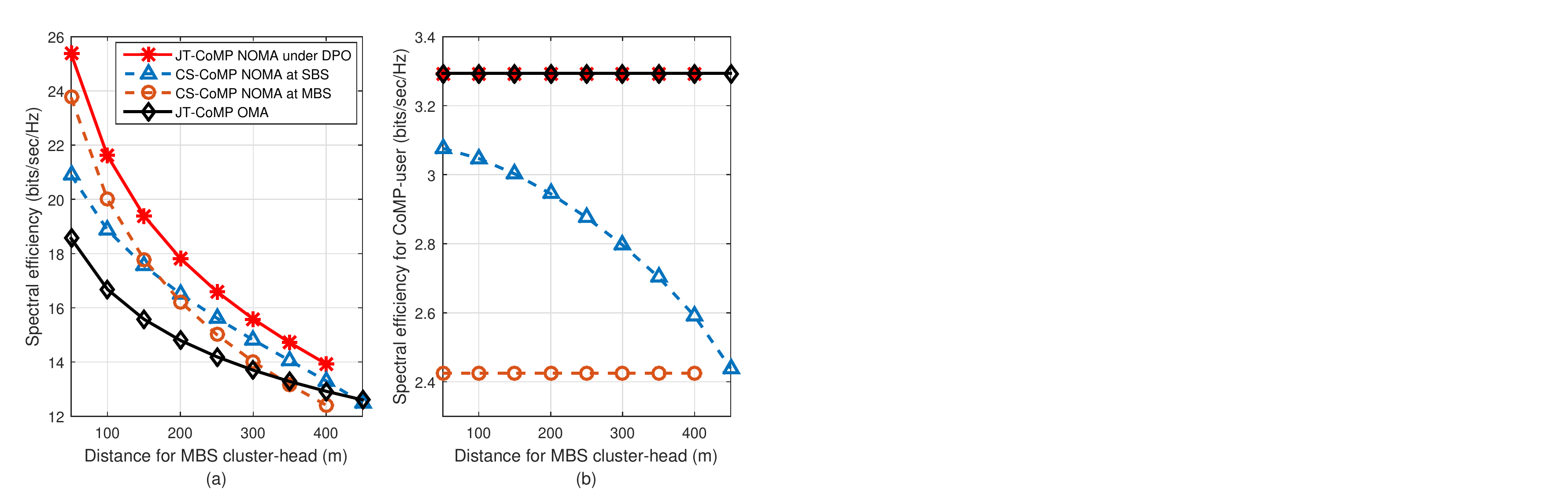}
	\caption{Spectral efficiency for JT-CoMP-NOMA model $2$:$2$:$1$, corresponding CS-CoMP-NOMA and JT-CoMP-OMA. (a) SE over CoMP-set, (b) SE for CoMP-UE. Non-CoMP-UE of SBS and CoMP-UE are at a distance of $50$ m, and $150$ m, respectively, from SBS.}
	\label{fig:vi}
 \end{center}
\end{figure}

The EE performance of the proposed CoMP-NOMA system under both JPO and DPO approaches and its comparison with the CoMP-OMA system are demonstrated in Fig. \ref{fig:v}. 
As discussed in Section VI, under a JT-CoMP NOMA system, all the CoMP-BSs  need to individually satisfy the rate requirements for CoMP-users whose SINRs depend on the ratio of desired signal power and INUI power. As a result, the EE performance gain for JT-CoMP NOMA in Fig. \ref{fig:v} (index view) does not show significant improvement over JT-CoMP-OMA. However, in the case of the CS-CoMP-NOMA discussed in Section VI.A, the EE gain dramatically increases  when the CoMP-user is served by an SBS. The reasons include the high channel gain for CoMP-UE with the SBS, low transmit power budget of  an SBS, and high transmit power of eNB. When the CoMP-UE is served by the SBS, it  uses its full power budget while the eNB only uses the minimum power that it needs to fulfill the non-CoMP-UEs' minimum rate requirements. Fig. \ref{fig:vi}(a) shows the SE for the proposed CS-CoMP-NOMA CoMP-set in comparison with JT-CoMP OMA and JT-CoMP-NOMA model $2$:$2$:$1$.

Under CS-CoMP, all the CoMP-BSs except one, which serves the CoMP-UE,  utilize power control, thus the offset ICI in \eqref{EE_1} could not be utilized. 
As a result, all individual user's achievable data rate under CS-JT-CoMP transmission may not meet the rate requirement that could be obtained in JT-CoMP-OMA. In particular, a CoMP-UE's date rate decreases as the channel gains for the non-CoMP-UEs served by the eNB decease, which is shown in Fig. \ref{fig:vi}(b). This can be intuitively explained as follows: as the distance of a non-CoMP-UE served by the eNB increases with respect to the eNB itself, then the effect of channel gain in the SINR formula  decreases and thus the required amount of transmit power increases. Increase in the transmit power for non-CoMP-UEs served by the eNB significantly increases ICI to CoMP-UEs and thus the SINR (and hence the achievable data rate) decreases. The ICI for non-CoMP-UEs served by the SBS also increases; however, this impact can be small if its channel gain is sufficiently high. 


\section{Conclusion}
We have modeled and analyzed the problem of dynamic power allocation for downlink CoMP-NOMA  in a two-tier HetNet. For a CoMP-set, we have formulated the optimal power allocation for downlink CoMP-NOMA  as a joint power optimization problem among the coordinating BSs. Due to the high computational complexity associated with the proposed joint power optimization approach, 
a low-complexity distributed power optimization method has been proposed for each coordinating BS. The optimal power allocation solution for each coordinating BS under the distributed approach is independent of the solutions at other coordinating BSs. We have also derived the necessary conditions under which the distributed power optimization solution becomes a valid solution for the joint power optimization problem. Finally, through rigorous numerical analysis, we have demonstrated the gain in spectral efficiency for the proposed CoMP-NOMA system in comparison with their CoMP-OMA counterparts. The performance gap between the proposed joint power optimization and distributed power optimization approaches has been considered as well. Insights on the energy efficiency performance of the proposed CoMP-NOMA models have been also discussed. The proposed dynamic power allocation model for CoMP-NOMA in a two-tier HetNet is also valid for any multi-cell downlink system.



\begin{appendices}
\renewcommand{\theequation}{A.\arabic{equation}}
\setcounter{equation}{0}
\section{proof of lemma 1}
To prove \textbf{Lemma 1} we need to prove that the Hessian matrix for an $M$-user, $M\geq 2$, downlink NOMA sum-rate is negative definite. A matrix is negative definite if all its $m$-th order, $ \forall m \in M$, leading principal minors alternate in sign, starting from negative, i.e., the minors of odd-numbered order are negative and those of even-numbered order are positive.

As an example, consider a $2$-user downlink NOMA. Then the sum-rate formula over unit NOMA resource block is
\begin{align}
R_{sum}^{(2)} = \log_2\Big(1+\frac{p_1 \gamma_1}{p_2\gamma_1 + 1}\Big) + \log_2\Big(1 + p_2\gamma_2\Big)
\label{L1_1}
\end{align}
where $\gamma_2 > \gamma_1$. The Hessian matrix of \eqref{L1_1} can be written as 
\begin{align*}
\bm{\mathbf{H}}^{(2)} = 
\begin{bmatrix}
\frac{\partial^2 R_{sum}^{(2)}}{\partial p_1^2} & \frac{\partial^2 R_{sum}^{(2)}}{\partial p_1\partial p_2} \\
\frac{\partial^2 R_{sum}^{(2)}}{\partial p_2\partial p_1} & \frac{\partial^2 R_{sum}^{(2)}}{\partial p_2^2}
\end{bmatrix} = 
\begin{bmatrix}
-\pi_2 & -\pi_2 \\
-\pi_2 & -\pi_2 + \pi_2^\prime - \phi_2
\end{bmatrix}
\end{align*}
where $\pi_2 = \left(\frac{\gamma_1}{(p_1 + p_2)\gamma_1 +1}\right)^2$, $\pi_2^\prime = \left(\frac{\gamma_1}{p_2\gamma_1 +1}\right)^2$ and $\phi_2 = \left(\frac{\gamma_2}{p_2\gamma_2 +1}\right)^2$. Then the first-order principle minor of $\bm{\mathbf{H}}^{(2)}$, denoted as $\det{\begin{vmatrix}{\bm{\mathbf{H}}_1^{(2)}}\end{vmatrix}}$, can be written as
\begin{align}
\det{\begin{vmatrix}{\bm{\mathbf{H}}_1^{(2)}}\end{vmatrix}} = -\pi_2.
\label{L1_2_1}
\end{align}
The second-order principle minor of $\bm{\mathbf{H}}^{(2)}$, denoted as $\det{\begin{vmatrix}{\bm{\mathbf{H}}_2^{(2)}}\end{vmatrix}}$, can be written as
\begin{align}
\det{\begin{vmatrix}{\bm{\mathbf{H}}_2^{(2)}}\end{vmatrix}} &= \pi_2(\phi_2 - \pi_2^\prime)  \nonumber \\
&=\pi_2\left(\frac{(2p_2\gamma_1\gamma_2 + \gamma_1 + \gamma_2)(\gamma_2 - \gamma_1)}{(p_2\gamma_1 + 1)^2(p_2\gamma_2 + 1)^2}\right).
\label{L1_2_2}
\end{align}
Since $\pi_2, \pi_2^\prime$ and $\phi_2$ all are positive and $\gamma_2 >\gamma_1$, the first-order principle minor of $\mathrm{\bm{\mathbf{H}}^{(2)}}$ is negative and those for second order is positive. Therefore, the the sum-rate of $2$-user downlink NOMA of ascending channel gain based SIC ordering is a strictly concave function of transmit power.

Now, consider a $3$-user downlink NOMA of ascending channel gain based SIC ordering, then the sum-rate over unit NOMA resource block is 
\begin{align}
R_{sum}^{(3)} = \sum\limits_{i = 1}^{3}\log_2\Bigg(1+\frac{p_i \gamma_i}{\sum\limits_{j = i+1}^{3} p_j\gamma_i + 1}\Bigg)
\label{L1_3}
\end{align}
where $\gamma_3 > \gamma_2 > \gamma_1$. The Hessian matrix of \eqref{L1_3} can be written as 
\begin{align*}
\mathrm{\bm{\mathbf{H}}^{(3)}} = 
\begin{bmatrix}
-\pi_3 & -\pi_3  & -\pi_3\\
-\pi_3 & -\pi_3 + \pi_3^\prime - \phi_3 & -\pi_3 + \pi_3^\prime - \phi_3 \\
-\pi_3 & -\pi_3 + \pi_3^\prime - \phi_3 & -\pi_3 + \pi_3^\prime - \phi_3 +\phi_3^\prime - \psi_3
\end{bmatrix}
\end{align*}
where $\pi_3 = \left(\frac{\gamma_1}{(p_1 + p_2 + p_3)\gamma_1 +1}\right)^2$, $\pi_3^\prime = \left(\frac{\gamma_1}{(p_2 + p_3)\gamma_1 +1}\right)^2$, $\phi_3 = \left(\frac{\gamma_2}{(p_2+p_3)\gamma_2 +1}\right)^2$, $\phi_3^\prime = \left(\frac{\gamma_2}{p_3\gamma_2 +1}\right)^2$ and $\psi_3 = \left(\frac{\gamma_3}{p_3\gamma_2 +1}\right)^2$. Then the first-order principle minor of $\mathrm{\bm{\mathbf{H}}^{(3)}}$, denoted as $\det{\begin{vmatrix}{\bm{\mathbf{H}}_1^{(3)}}\end{vmatrix}}$, can be written as
\begin{align}
\det{\begin{vmatrix}{\bm{\mathbf{H}}_1^{(3)}}\end{vmatrix}} = -\pi_3.
\label{L1_2_1}
\end{align}
The second-order principle minor of $\mathrm{\bm{\mathbf{H}}^{(3)}}$, denoted as $\det{\begin{vmatrix}{\bm{\mathbf{H}}_2^{(3)}}\end{vmatrix}}$, can be written as
\begin{align}
\det{\begin{vmatrix}{\bm{\mathbf{H}}_2^{(3)}}\end{vmatrix}} &= \pi_3(\phi_3 - \pi_3^\prime)  =\pi_3\left(\frac{(2(p_2+p_3)\gamma_1\gamma_2 + \gamma_1 + \gamma_2)(\gamma_2 - \gamma_1)}{((p_2+p_3)\gamma_1 + 1)^2((p_2+p_3)\gamma_2 + 1)^2}\right).
\label{L1_2_2}
\end{align}
Also, the third-order principle minor of $\mathrm{\bm{\mathbf{H}}^{(3)}}$, denoted as $\det{\begin{vmatrix}{\bm{\mathbf{H}}_3^{(3)}}\end{vmatrix}}$, can be written as
\begin{align}
&\det{\begin{vmatrix}{\bm{\mathbf{H}}_3^{(3)}}\end{vmatrix}} = -\pi_3(\phi_3 - \pi_3^\prime)(\psi_3 - \phi_3^\prime)  \nonumber \\
&= -\pi_3\left(\frac{(2(p_2+p_3)\gamma_1\gamma_2 + \gamma_1 + \gamma_2)(\gamma_2 - \gamma_1)}{((p_2+p_3)\gamma_1 + 1)^2((p_2+p_3)\gamma_2 + 1)^2}\right) \left(\frac{(2p_3\gamma_2\gamma_3 + \gamma_2 + \gamma_3)(\gamma_3 - \gamma_2)}{(p_3\gamma_2 + 1)^2(p_3\gamma_3 + 1)^2}\right).
\label{L1_2_2}
\end{align}
Since $\pi_2, \pi_2^\prime, \phi_2, \phi_2^\prime$ and $\psi_2$ all are positive and $\gamma_2 > \gamma_2 >\gamma_1$, the first-order and third-order principle minor of $\mathrm{\bm{\mathbf{H}}^{(3)}}$ are negative while those for second-order are positive. Therefore, the the sum-rate of $3$-user downlink NOMA of ascending channel gain based SIC ordering is also a strictly concave function of transmit power. 

The proof can be easily expanded for an $M$-user ($M \geq 2$) downlink NOMA system of ascending channel gain-based SIC decoding order. Using induction, for an $M$-user downlink NOMA system, the leading principal minors can be written as
\begin{align}
\det{\begin{vmatrix}{\bm{\mathbf{H}}_1^{(M)}}\end{vmatrix}}& = -\pi_M \\
\det{\begin{vmatrix}{\bm{\mathbf{H}}_2^{(M)}}\end{vmatrix}}& = \pi_M(\phi_M - \pi_M^\prime) \\
\det{\begin{vmatrix}{\bm{\mathbf{H}}_3^{(M)}}\end{vmatrix}}& = -\pi_M(\phi_M - \pi_M^\prime)(\psi_M - \phi_M^\prime) \\
\vdots \qquad\\
\det{\begin{vmatrix}{\bm{\mathbf{H}}_M^{(M)}}\end{vmatrix}}& = \Delta \pi_M(\phi_M - \pi_M^\prime)(\psi_M - \phi_M^\prime)\cdots(\varpi_M - \eta_M^\prime)
\end{align}
where $\Delta$ represents ($+$) sign for even values of $M$ and ($-$) sign for odd values of $M$, and
\begin{align*} \pi_M &= \left[\frac{\gamma_1}{\sum\limits_{m = 1}^M p_m\gamma_1 +1}\right]^2, \enspace  
\pi_M^\prime = \left[\frac{\gamma_1}{\sum\limits_{m = 2}^M p_m\gamma_1 +1}\right]^2, \enspace  
\phi_M = \left[\frac{\gamma_2}{\sum\limits_{m = 2}^M p_m\gamma_2 +1}\right]^2, \\
\phi_M^\prime &= \left[\frac{\gamma_2}{\sum\limits_{m = 3}^M p_m\gamma_2 +1}\right]^2, \enspace
\psi_M = \left[\frac{\gamma_3}{\sum\limits_{m = 3}^M p_m\gamma_3 +1}\right]^2, \qquad \quad
\cdots \\
\end{align*}
\begin{align*}
\eta_M^\prime &= \left[\frac{\gamma_{M-1}}{\sum\limits_{m = M-1}^M p_m\gamma_{M-1} +1}\right]^2, 
\varpi_M = \left[\frac{\gamma_M}{\sum\limits_{m = M}^M p_m\gamma_M +1}\right]^2.
\end{align*}

This proves \textbf{Lemma 1}.

\section{Proof of Lemma~4}
\renewcommand{\theequation}{C.\arabic{equation}}
\setcounter{equation}{0}
To prove \textbf{Lemma~4}, the desired signal power for a CoMP-user is derived as a function of inter-NOMA-user interference power, noise power, and required data rate under joint optimization approach \eqref{mcn_op_1} and distributed optimization approach \eqref{s_mcn_op_1}. Under the distributed approach, each CoMP-BS independently allocates power, and all CoMP-BSs simultaneously transmit to a CoMP-UE, therefore, the resultant desired signal power for CoMP-UEs will be the summation of the individual ones. Finally,  this Lemma can be proved by ensuring that the signal power (normalized with respect to noise power) under the distributed approach is greater than or equal to that under the joint approach.

\underline{\em Desired Signal Power for Joint Optimization Approach}:
Under CoMP-set $y$, the NOMA cluster served by the eNB contains $\Phi_{y,m}$ non-CoMP-UEs and $\Phi_{y,ms_x}$ CoMP-UEs, and the NOMA cluster served by SBS $x$ contains $\Phi_{y,s_x}$ non-CoMP-UEs and $\Phi_{y,ms_x}$ CoMP-users. For CoMP-set $y$, let $\Psi_{y,m} = \Phi_{y,m} + \Phi_{y,ms_x}$ denote the total number of users in the NOMA cluster served by the  eNB and  $\Psi_{y,s_x} = \Phi_{y,s_x} + \Phi_{y,ms_x}$ is the total number of users in the NOMA cluster served by SBS $x$. Then the rate constraint for CoMP-UE $k\in \{1,\cdots, \Phi_{y,ms_x}\}$ under the joint optimization approach in \eqref{mcn_op_1} can be written as 
\begin{align}
R_{k} = \log_2\Bigg(1+ \frac{p_{k}^{(m)}\gamma_{k}^{(m)} + p_{k}^{(s_x)}\gamma_{k}^{(s_x)}}{\sum\limits_{i^\prime = k+1}^{\Psi_{y,m}} p_{i^\prime}^{(m)}\gamma_{k}^{(m)} + \sum\limits_{j^\prime = k+1}^{\Psi_{y,s_x}} p_{j^\prime}^{(s_x)}\gamma_{k}^{(s_x)} + 1}\Bigg) \geq R_k^\prime.
\label{lem5_2}
\end{align}
On the other hand, the CoMP-UEs' rate requirement $R_k^\prime$ in terms of their achievable rates under OMA can be expressed as
\begin{align}
R_{k}^\prime = \beta_{k}\log_2\left(1+\frac{\alpha_k^{(m)} p_{t}^{(m)}\gamma_{k}^{(m)}}{\beta_{k}} + \frac{\alpha_k^{(s_x)} p_{t}^{(s_x)}\gamma_{k}^{(s_x)}}{\beta_{k}}\right)
\label{lem5_3}
\end{align}
where $\beta_{k} = [0,1]$ and $\alpha_{k}^{(n)} = [0,1], \forall n \in \{m, s_x\}$ define the spectrum allocation and power allocation coefficients for CoMP-UE $k$ under OMA. 
Combining \eqref{lem5_2} and \eqref{lem5_3} yields
\begin{align}
\frac{\left(p_t^{(m)}-\sum\limits_{i^\prime = 1}^{k-1} p_{i^\prime}^{(m)}\right)\gamma_{k}^{(m)} + \left(p_t^{(s_x)}- \sum\limits_{j^\prime = 1}^{k-1} p_{j^\prime}^{(s_x)}\right)\gamma_{k}^{(s_x)} +1}{\left(p_t^{(m)}-\sum\limits_{i^\prime = 1}^{k} p_{i^\prime}^{(m)}\right)\gamma_{k}^{(m)} + \left(p_t^{(s_x)}- \sum\limits_{j^\prime = 1}^{k} p_{j^\prime}^{(s_x)}\right)\gamma_{k}^{(s_x)} +1} \geq \left(1+\frac{\alpha_k^{(m)} p_{t}^{(m)}\gamma_{k}^{(m)}}{\beta_{k}} + \frac{\alpha_k^{(s_x)} p_{t}^{(s_x)}\gamma_{k}^{(s_x)}}{\beta_{k}}\right)^{\beta_{k}}
\label{lem5_4}
\end{align}
where $p_t^{(m)} = \sum\limits_{i^\prime = 1}^{\Psi_{y,m}} p_{i^\prime}^{(m)}$ and $p_t^{(s_x)} = \sum\limits_{j^\prime = 1}^{\Psi_{y,s_x}} p_{j^\prime}^{(s_x)}$ are the NOMA power budget at MBS and $x$-th SBS end, respectively. By letting $\Gamma^{(m)} = p_t^{(m)}\gamma_{k}^{(m)}$, $\Gamma^{(s_x)} = p_t^{(s_x)}\gamma_{k}^{(s_x)}$, $\Gamma_{k-1}^{(m)} = \sum\limits_{i^\prime = 1}^{k-1} p_{i^\prime}^{(m)}\gamma_{k}^{(m)}$, $\Gamma_{k-1}^{(s_x)} = \sum\limits_{j^\prime = 1}^{k} p_{j^\prime}^{(s_x)}\gamma_{k}^{(s_x)}$, $\Gamma = \left(1+\frac{\alpha_k^{(m)} p_{t}^{(m)}\gamma_{k}^{(m)}}{\beta_{k}} + \frac{\alpha_k^{(s_x)} p_{t}^{(s_x)}\gamma_{k}^{(s_x)}}{\beta_{k}}\right)^{\beta_{k}}$, and after some algebraic manipulations, we can write \eqref{lem5_4} as 
\begin{align}
p_{k}^{(m)}\gamma_{k}^{(m)} + p_{k}^{(s_x)}\gamma_{k}^{(s_x)} \geq \left( 1 - \frac{1}{\Gamma} \right)\bigg(\Gamma^{(m)} + \Gamma^{(s_x)} + 1 -\Gamma_{k-1}^{(m)}-\Gamma_{k-1}^{(s_x)} \bigg).
\label{lem5_5}
\end{align}
The inequality \eqref{lem5_5} represents the desired signal power for CoMP-UE $k$ as a function of inter-NOMA-user interference power, noise power, and required data rate under joint optimization approach \eqref{mcn_op_1}.


\underline{\em Desired Signal Power for Distributed Optimization Approach}:
Consider the rate constraint for a CoMP-UE $k$ under the distributed optimization approach in \eqref{mcn_op_1}. At the MBS end, i.e., for $n = m$, the rate constraint for CoMP-UE $k$ could be written as
\begin{align}
R_{k}^{(m)} \geq \log_2\left(1+ \frac{\hat{p}_{k}^{(m)}\gamma_{k}^{(m)} }{\sum\limits_{i^\prime = k+1}^{\Psi_{y,m}} \hat{p}_{i^\prime}^{(m)}\gamma_{k}^{(m)} + 1}\right) \geq R_k^\prime
\label{lem5_2_2}
\end{align}
where $\hat{p}_k^{(m)}$ represents the power that the eNB allocates to CoMP-UE $k$ under the distributed optimization approach. 
Combining \eqref{lem5_2_2} and \eqref{lem5_3} gives
\begin{align}
\frac{p_t^{(m)}-\sum\limits_{i^\prime = 1}^{k-1} \hat{p}_{i^\prime}^{(m)}\gamma_{k}^{(m)} + 1}{p_t^{(m)}-\sum\limits_{i^\prime = 1}^{k} \hat{p}_{i^\prime}^{(m)}\gamma_{k}^{(m)} +1} \geq \Gamma. \nonumber
\end{align}
Therefore,
\begin{align}
\hat{p}_{k}^{(m)}\gamma_{k}^{(m)} \geq \left( 1 - \frac{1}{\Gamma} \right)\bigg(\Gamma^{(m)} + 1 - \hat{\Gamma}_{k-1}^{(m)} \bigg)
\label{lem5_5_5}
\end{align}
where $\hat{\Gamma}_{k-1}^{(m)} = \sum_{i^\prime = 1}^{k-1} \hat{p}_{i^\prime}^{(m)}\gamma_{k}^{(m)}$.  Similar to eNB, \eqref{lem5_5_5} can be derived for SBS $x$, i.e., for $n = s_x$, as follows:
\begin{align}
\hat{p}_{k}^{(m)}\gamma_{k}^{(m)} \geq \left( 1 - \frac{1}{\Gamma} \right)\bigg(\Gamma^{(m)} + 1 - \hat{\Gamma}_{k-1}^{(m)} \bigg)
\label{lem5_5_6}
\end{align}
where $\hat{\Gamma}_{k-1}^{(s_x)} = \sum_{j^\prime = 1}^{k-1} \hat{p}_{j^\prime}^{(s_x)}\gamma_{k}^{(s_x)}$. \eqref{lem5_5_5} and \eqref{lem5_5_6} show the desired signal power for CoMP-user $k$ from MBS and SBS $x$, respectively, which are expressed in terms of their respective inter-NOMA-user interference power, noise power and required data rate for CoMP-UE $k$. Therefore, the resultant desired signal power for CoMP-UE $x$ under the distributed approach can be expressed as
\begin{align}
\hat{p}_{k}^{(m)}\gamma_{k}^{(m)} + \hat{p}_{k}^{(s_x)}\gamma_{k}^{(s_x)} \geq \left( 1 - \frac{1}{\Gamma} \right)\bigg(\Gamma^{(m)} + \Gamma^{(s_x)} + 2 - \hat{\Gamma}_{k-1}^{(m)}-\hat{\Gamma}_{k-1}^{(s_x)} \bigg).
\label{lem5_6_6}
\end{align}

\underline{\em Comparison Between the Desired Signal Powers}:
To prove \textbf{Lemma~4}, we need to prove that,
the  desired signal power under the distributed optimization approach is greater than or equal to the  desired signal power under the joint optimization approach, i.e., 
under equality condition of \eqref{lem5_5} and \eqref{lem5_6_6},
\begin{align}
\hat{p}_{k}^{(m)}\gamma_{k}^{(m)} + \hat{p}_{k}^{(s_x)}\gamma_{k}^{(s_x)} \geq p_{k}^{(m)}\gamma_{k}^{(m)} + p_{k}^{(s_x)}\gamma_{k}^{(s_x)}.
\label{lem5_p_1}
\end{align}

\textit{Case 1:} For $k=1$, we get $\Gamma_{k-1}^{(m)} =\Gamma_{k-1}^{(s_x)}=\hat{\Gamma}_{k-1}^{(m)}=\hat{\Gamma}_{k-1}^{(s_x)} = 0$. Substituting the result into \eqref{lem5_4} and \eqref{lem5_6_6} gives 
\begin{align*}
&p_{1}^{(m)}\gamma_{1}^{(m)} + p_{1}^{(s_x)}\gamma_{1}^{(s_x)} = \left( 1 - \frac{1}{\Gamma} \right)\bigg(\Gamma^{(m)} + \Gamma^{(s_x)} + 1\bigg) \\
&\hat{p}_{1}^{(m)}\gamma_{1}^{(m)} + \hat{p}_{1}^{(s_x)}\gamma_{1}^{(s_x)} = \left( 1 - \frac{1}{\Gamma} \right)\bigg(\Gamma^{(m)} + \Gamma^{(s_x)} + 2\bigg).
\end{align*}
Therefore, \eqref{lem5_p_1} is valid for $k = 1$. 

\textit{Case 2:} For $k=2$, $\Gamma_{k-1}^{(m)} =p_{1}^{(m)}\gamma_{1}^{(m)}$, $ \Gamma_{k-1}^{(s_x)}= p_{1}^{(s_x)}\gamma_{1}^{(s_x)}$, $\hat{\Gamma}_{k-1}^{(m)}= \hat{p}_{1}^{(m)}\gamma_{1}^{(m)}$, $\hat{\Gamma}_{k-1}^{(s_x)} = \hat{p}_{1}^{(s_x)}\gamma_{1}^{(s_x)}$. Substituting the result into \eqref{lem5_4} and \eqref{lem5_6_6} gives
\begin{align*}
&p_{2}^{(m)}\gamma_{2}^{(m)} + p_{2}^{(s_x)}\gamma_{2}^{(s_x)} = \left( 1 - \frac{1}{\Gamma} \right)\bigg(\frac{\Gamma^{(m)} + \Gamma^{(s_x)} + 1}{\Gamma}\bigg) \\
&\hat{p}_{2}^{(m)}\gamma_{2}^{(m)} + \hat{p}_{2}^{(s_x)}\gamma_{2}^{(s_x)} = \left( 1 - \frac{1}{\Gamma} \right)\bigg(\frac{\Gamma^{(m)} + \Gamma^{(s_x)} + 2}{\Gamma}\bigg).
\end{align*}
Thus \eqref{lem5_p_1} is also valid for $k = 2$. Finally, for any $k \geq 1$, 
\begin{align*}
&p_{k}^{(m)}\gamma_{k}^{(m)} + p_{k}^{(s_x)}\gamma_{k}^{(s_x)} = \left( 1 - \frac{1}{\Gamma} \right)\bigg(\frac{\Gamma^{(m)} + \Gamma^{(s_x)} + 1}{\Gamma^{k-1}}\bigg), \\
&\hat{p}_{k}^{(m)}\gamma_{k}^{(m)} + \hat{p}_{k}^{(s_x)}\gamma_{k}^{(s_x)} = \left( 1 - \frac{1}{\Gamma} \right)\bigg(\frac{\Gamma^{(m)} + \Gamma^{(s_x)} + 2}{\Gamma^{k-1}}\bigg).
\end{align*}
Moreover, from \eqref{lem5_5}, \eqref{lem5_6_6} and \eqref{lem5_p_1}, it can be noted that if a CoMP-UE's noise power is divided by the number of CoMP-BSs, then the CoMP-UE's data rate under JPO and DPO will be similar, which is stated in \textbf{Lemma~4}.

\end{appendices}

\bibliographystyle{IEEE}

\end{document}